\documentclass[a4paper,UKenglish,cleveref, autoref, thm-restate]{lipics-v2021}
\usepackage{cite}
\usepackage{graphicx}
\usepackage{amsmath}

\usepackage{thmtools}
\usepackage{thm-restate}
\usepackage{hyperref}
\usepackage{cleveref}

\usepackage[createShortEnv]{proof-at-the-end}

\usepackage{xcolor}
\usepackage{tabularx}
\usepackage{verbatimbox}
\usepackage{xspace}

\usepackage{prettyref}
\usepackage{wrapfig}
\usepackage{cutwin}
\usepackage{mathtools}
\usepackage{adjustbox}
\usepackage{float}
\usepackage{longtable}
\usepackage{stackengine}
\usepackage{xfrac}

\newrefformat{sec}{Section\,\ref{#1}}
\newrefformat{def}{Def.\,\ref{#1}}
\newrefformat{thm}{Theorem\,\ref{#1}}
\newrefformat{prop}{Proposition\,\ref{#1}}
\newrefformat{lem}{Lemma\,\ref{#1}}
\newrefformat{cor}{Corollary\,\ref{#1}}
\newrefformat{ex}{Example\,\ref{#1}}
\newrefformat{tab}{Table\,\ref{#1}}
\newrefformat{fig}{Fig.\,\ref{#1}}
\newrefformat{app}{Appendix\,\ref{#1}}
\usepackage{stmaryrd}
\usepackage{amssymb}
\newcommand{\logic}{\text{\upshape\textsf{d{\kern-0.05em}G{\kern-0.15em}$\mathcal{L}_{3}$}}\xspace}
\usepackage{bussproofs}
\usepackage{xparse}
\newcommand{\dDL}[1][]
{\text{\upshape\textsf{d{\kern-0.1em}$\mathcal{L}$}}\xspace}
\newcommand{\C}{\complement}
\newcommand{\s}{\mathcal{S}}

\newcommand{\spl}{\Diamond\text{spl}}
\newcommand{\splb}{\Box\text{spl}}
\newcommand{\aspl}{\Diamond\text{aspl}}
\newcommand{\asplb}{\Box\text{aspl}}
\newcommand{\dco}{\Diamond^\co}
\newcommand{\dcob}{\Box^\co}
\newcommand{\dlor}{\Diamond\lor}
\newcommand{\bland}{\Box\land}
\newcommand{\db}{\Box \cdot}
\newcommand{\sdb}{\Box\{\cdot\}}
\newcommand{\sas}{\Diamond\{\cdot\}\lstrike := \rstrike}
\newcommand{\as}{\Diamond\lstrike := \rstrike}
\newcommand{\asb}{\Box\lstrike := \rstrike}
\newcommand{\scon}{\Diamond\{\cdot\}\lstrike' \rstrike}
\newcommand{\con}{\Diamond\lstrike' \rstrike}
\newcommand{\conb}{\Box\lstrike' \rstrike}
\newcommand{\sconn}{\Diamond\{\cdot\}\lstrike' \rstrike^\C}
\newcommand{\conn}{\Diamond\lstrike' \rstrike^\C}
\newcommand{\connb}{\Box\lstrike' \rstrike^\C}
\newcommand{\ste}{\Diamond\{\cdot\}\lstrike? \rstrike}
\newcommand{\te}{\Diamond\lstrike? \rstrike}

\newcommand{\sten}{\Diamond\{\cdot\}\lstrike? \rstrike^\C}
\newcommand{\ten}{\Diamond\lstrike? \rstrike^\C}
\newcommand{\tenb}{\Box\lstrike? \rstrike^\C}
\newcommand{\sch}{\Diamond\{\cdot\}\lstrike \cup \rstrike}
\newcommand{\ch}{\Diamond\lstrike \cup \rstrike}
\newcommand{\chb}{\Box\lstrike \cup \rstrike}
\newcommand{\schn}{\Diamond\{\cdot\}\lstrike \cup \rstrike^\C}
\newcommand{\chn}{\Diamond\lstrike \cup \rstrike^\C}
\newcommand{\chnb}{\Box\lstrike \cup \rstrike^\C}
\newcommand{\sseq}{\Diamond\{\cdot\}\lstrike {;} \rstrike}
\newcommand{\seq}{\Diamond\lstrike {;} \rstrike}
\newcommand{\seqb}{\Box\lstrike {;} \rstrike}
\newcommand{\srep}{\Diamond\{\cdot\}\lstrike^*\rstrike}
\newcommand{\rep}{\Diamond\lstrike^*\rstrike}
\newcommand{\repb}{\Box\lstrike^*\rstrike}
\newcommand{\srepn}{\Diamond\{\cdot\}\lstrike^*\rstrike^\C}
\newcommand{\repn}{\Diamond\lstrike^*\rstrike^\C}
\newcommand{\repnb}{\Box\lstrike^*\rstrike^\C}
\newcommand{\sFP}{\Diamond\{\cdot\}\text{FP}}
\newcommand{\FP}{\Diamond\text{FP}}
\newcommand{\FPb}{\Box\text{FP}}
\newcommand{\sind}{\Diamond\{\cdot\}\text{ind}}
\newcommand{\ind}{\Diamond\text{ind}}
\newcommand{\indb}{\Box\text{ind}}
\newcommand{\sMon}{\Diamond\{\cdot\}\text{M}}
\newcommand{\Mon}{\Diamond\text{M}}
\newcommand{\Monb}{\Box\text{M}}

\newcommand{\xir}{\Diamond\Xi}
\newcommand{\xirb}{\Box\Xi}
\newcommand{\co}{\mathcal{C}}

\newcommand{\win}[1]{\llbracket #1 \rrbracket}

\newcommand{\dGL}[1][]
{\text{\upshape\textsf{d{\kern-0.05em}G{\kern-0.15em}$\mathcal{L}$}}\xspace}
\newcommand{\dGLsc}{\text{\upshape\textsf{d{\kern-0.05em}G{\kern-0.15em}$\mathcal{L}_{sc}$}}\xspace}
\newcommand{\QdL}[1][]
{\text{\upshape\textsf{$\mathcal{Q}$d{\kern-0.05em}$\mathcal{L}$}}\xspace}

\newcommand{\dLchp}[1][]
{\text{\upshape\textsf{d{\kern-0.05em}$\mathcal{L}_{\text{CHP}}$}}\xspace}

\newcommand{\lstrike}{\mathopen{\mathrel{\langle}\joinrel\mathrel{[}}}
\newcommand{\rstrike}{\mathclose{\mathrel{]}\joinrel\mathrel{\rangle}}}

\newcommand{\mo}[2]{\lstrike #1 \rstrike (#2)}

\newcommand{\proj}[3]{\Xi(#1,#2,#3)}

\newcommand{\dia}[2]{\Diamond #1 : #2}
\newcommand{\bx}[2]{\square #1 : #2}
\newcommand{\angel}[2]{\langle #1 \rangle #2}
\newcommand{\demon}[2]{[#1]#2}
\newcommand{\dc}{^{-c}}
\newcommand{\winrA}[2]{\varsigma_{#1}(#2)}
\newcommand{\winrD}[2]{\delta_{#1}(#2)}

\usepackage{ifpdf}
\ifpdf
\fi

\nolinenumbers
\title{Three-player Differential Game Logic}
\author{Julia Butte}{Department of Informatics, Karlsruhe Institute of Technology, Karlsruhe, Germany}{julia.butte@kit.edu}{https://orcid.org/0009-0003-5066-8412}{}

\author{Andr\'e Platzer}{Department of Informatics, Karlsruhe Institute of Technology, Karlsruhe, Germany}{platzer@kit.edu}{https://orcid.org/0000-0001-7238-5710}{}

\authorrunning{J. Butte and A. Platzer}

\Copyright{Julia Butte and Andr\'e Platzer}



\acknowledgements{This work was funded by an Alexander von Humboldt Professorship and by KiKIT, the Pilot Program for Core-Informatics at the KIT of the Helmholtz association.}

\begin{document}
\maketitle
\begin{abstract}
    This paper introduces the \emph{three-player differential game logic} \logic, which enables the verification of hybrid games of discrete and differential equation dynamics with three players who may or may not form coalitions.
    Each player has an individual goal they try to fulfill, so the game becomes non-zero-sum when the goals of the players overlap.
    This is how \logic can to verify complex situations involving multiple players, taking into account their coalitional power.
    \logic excels at verifying games where players share safety objectives but otherwise pursue different goals, so that they may or may not collaborate.
    In this case, zero-sum assumptions lead to overly conservative results by neglecting the potential of coordination amongst the players.
    In this paper, the syntax and semantics of \logic are presented and crucial properties of \logic are proved.
    A sound and relatively complete proof calculus for \logic is introduced and the use of \logic is illustrated in a canonical example.
\end{abstract}

\section{Introduction}

Cyber-physical systems (CPS) like trains, planes, robots or autonomous cars \cite{DBLP:journals/tcad/KabraMP22,Lam2017, Taha2020,Jiang2010, DBLP:conf/cade/BriegerMP23} are all around us.
Their safety is critical in order to avoid damage to persons and goods.
Particularly challenging are cases where multiple CPSs with different goals interact with each other.
If all CPS agents are of one mind and cooperate to be controlled together, the resulting model is a hybrid system or single-player hybrid game \cite{Lin2021,Asarin2006,Platzer10,Alur15}.
If two different CPS agents with different goals are controlled separately and interact, the resulting model is a (two-player) hybrid game \cite{Henzinger19999, Platzer15,VLADIMEROU20116770,Mitchell2005}, because the players may reach decisions that interfere with one another.
The key challenge for truly multi-agent CPS, however, is when \emph{three} different CPS players interact, because that is when an entirely new genuinely \emph{multi}-agent dynamics arises: the three players interact with one another and may or may not form coalitions to work together and jointly achieve their respective goals.
Coalitions are a true three player challenge.
Coalitions are no challenge for two-player hybrid games, because coalitions reduce two-player hybrid games to single-player hybrid systems.
In three-player hybrid games, however, a proper hybrid game remains even when only two out of the three players decide to form a coalition, making the deliberation about whether or not and which coalition to form in a three-player hybrid game a genuinely novel challenge---the one that this paper addresses.
While hybrid games with even more than three players are of interest, too, the fundamental challenges already arise and can be solved in three-player hybrid games, which this paper, thus, focuses on for notational clarity.

This paper introduces the three-player differential game logic \logic for three-player hybrid games with discrete jumps and differential equation dynamics, which, importantly, may or may not form coalitions.
An axiomatization for \logic is presented and shown to be sound and complete relative to any differentially expressive logic.
In particular, it is shown that the logic of three-player games reduces to the logic of two-player games.
Thereby, this paper proof-theoretically equates
\[
\text{two-player hybrid games} = \text{three-player hybrid games} ~\text{(logically)}
\]
This is by no means clear \emph{a priori}, as three-player hybrid games introduce fundamentally new challenges that in and of themselves are beyond two-player hybrid games, like their ability to form coalitions at will.
This result is at first surprising and a consequence of the versatility of logic and its closure under all operators and modalities, not of just two-player hybrid games.
It demonstrates, however, that \logic is, indeed, suitably chosen as the logic of three-player hybrid games, and completeness establishes that \logic has all required reasoning principles to tame three-player hybrid games and their coalitional power.

The structure of the paper is as follows:
First, Section\,\ref{related_works} compares and relates \logic with existing work.
After that, Section\,\ref{preliminaries} briefly introduces the logic \dGL which will be used for the relative completeness proof of the calculus.
Then, Section\,\ref{dGL3} presents the syntax and semantics of \logic.
Section\,\ref{important_properties} establishes important properties of \logic. 
Finally, Section\,\ref{proof_calculus} introduces a sound and relatively complete proof calculus for \logic.

\section{Related Work}
\label{related_works}

\paragraph{General}
Three-player games are exceedingly challenging and have been considered by various communities.
Game theory provides fundamental definitions of three-player games and studies of their equilibria \cite{vonNeumannMorgenstern+2004,Riker_1967,Nash50} which lay the basis for our work.
Von Neumann and Morgenstern \cite{vonNeumannMorgenstern+2004} already detail the aspect that makes three-player games hard to handle and to understand:
Having three players allows the formation of coalitions amongst players which is not possible with only one or two players.
This adds an entirely new dynamic to the game, setting it distinctly apart from having only two players.
In terms of game theory, the games played in \logic are non-zero-sum sequential three-player games with perfect information and without transferable payoffs, but with uncountable action and state spaces and differential equations.
Like in subgame perfect equilibria \cite{Bielefeld1988}, \logic uses backward induction \cite{Aumann95} to determine whether a player can reach their goal, but it does so symbolically.
\logic has the advantage that no game trees are required for the backward induction which might have uncountably infinite breadth and unbounded depth due to the continuous dynamics.

For some special, non-hybrid cases of multiplayer games, rational synthesis  \cite{Fisman2010}, developed by  Fisman \emph{et al.}, can construct correct-by-construction solutions.
Unlike \logic, they cannot handle hybrid dynamics.
Another difference is that \logic is for verifying already existing systems, while rational synthesis is for constructing correct systems.

\paragraph{1-Player Games}
Hybrid systems without any game aspects have been investigated in some variety:
Differential dynamic logic enables verification of hybrid systems by extending dynamic logic with continuous dynamics \cite{DBLP:journals/jar/Platzer08}. 
\dLchp[] \cite{DBLP:conf/cade/BriegerMP23} addresses parallel hybrid systems and communication.
Hybrid Communicating Sequential Processes allows communication between processes \cite{Jiang2010}.

\paragraph{1\sfrac{1}{2}-Player Games}
Probabilistic hybrid systems can be seen as games with one and a half player where the half player represents probability and does not act with a strategy but randomly. 
Such systems are e.g. addressed by Lygeros \cite{Lygeros2006}, Koutsoukous and Rily \cite{Riley2008}, Platzer \cite{DBLP:conf/cade/Platzer11}, Fr{\"a}nzle \emph{et al.} \cite{Fraenzle2011}, and Zhao and Rozier \cite{Zhao2014}.
Unlike our work, none of these can handle game dynamics, not to mention three players.

\paragraph{2-Player Games}
Two-player games are investigated by different approaches:
Rectangular hybrid games \cite{Henzinger19999} by Henzinger \emph{\emph{et al.}} and STORMED games \cite{VLADIMEROU20116770} by Vladimerou \emph{\emph{et al.}} both consider hybrid games based on linear hybrid automata.
Unlike the hybrid games in \logic, these automata lack compositionality.
Mitchell \emph{et al.} \cite{Mitchell2005} numerically approximate the winning regions of two-player continuous games using Hamilton-Jacobi-Isaac partial differential equations. 
\logic instead, computes winning regions symbolically and accurately.
The logic \dGL \cite{Platzer15}, which \logic is based on, uses symbolic backward induction to determine the winning regions of two-player adversarial games.
The logic \dGLsc \cite{DBLP:conf/tableaux/ButteP25} can handle non-zero-sum games for two players.
None of these approaches addresses three players.

\paragraph{3-Player Games}
Tulenheimo and Venema \cite{tulenheimo2008} reduce propositional SAT checking to a three-player game.
In contrast to \logic, they are not reasoning \emph{about} three-player games, but they \emph{use} three-player games to determine the truth value of a propositional formula.

\paragraph{n-Player Games}
Alternating-time temporal logic (ATL) \cite{Alur2002} and stochastic game logic \cite{Baier2007} deal with non-hybrid multiplayer games with fixed coalitions.
Unlike in \logic, coalitions are fixed at the start which immediately breaks the game down to a fixed two-player game. 
As coalition formation is not considered, player goals are also not included, in contrast to \logic.

The logic of Agotnes \emph{et al.} \cite{agotnes2006} also handles non-hybrid games and takes player goals into account with a preference relation.
In contrast to \logic, preferability is only used to compare outcomes but not to exploit cooperation potential.
Unlike in \logic, the games are not composable, so proofs for subgames cannot be reused.

\QdL \cite{DBLP:conf/csl/Platzer10} can handle an arbitrary and changing number of players in a hybrid system but lacks \logic's game aspect, so all players have to work together in a multi-agent system. 

\section{Premliminaries}
\label{preliminaries}
In this section, the logic \dGL will be briefly introduced to enhance understanding of the paper.
\dGL is a logic for reasoning over two-player zero-sum games.
The two players are called \emph{Angel} and \emph{Demon}.
They play a hybrid game which is specified as part of the formula.
Now, syntax, semantics and the \dGL[] proof calculus will be briefly explained, based on \cite{Platzer2018}.

The syntax of \dGL is based on first-order logic.
Additionally, there are two modalities, the diamond modality $\angel{\alpha}{P}$ and the box modality $\demon{\alpha}{P}$.
The diamond modality means that Angel can win the game $\alpha$ by reaching her goal $P$ after the game has ended.
The box modality is its counterpart for Demon.
It means that Demon can win game $\alpha$ by achieving his goal $P$ after the game has ended.
The games are zero-sum, so the other player's goal is the opposite of the given goal and not specified separately.
Formally, the syntax is defined as:

\begin{definition}[\dGL syntax]
    Formulas of \dGL are defined by the grammar:
    \[\alpha, \beta ::= x:=e \mid ?Q \mid x'=f(x) \& Q \mid \alpha;\beta \mid \alpha \cup \beta \mid \alpha^* \mid \alpha^d\]
    \[P, Q ::= e \geq\tilde{e} \mid \lnot P \mid P \land Q \mid \forall x P \mid \exists x P \mid \angel{\alpha}{P} \mid \demon{\alpha}{P}\]
    where $\alpha$, $\beta$ are hybrid programs, $x$ is a variable, $e$, $\tilde{e}$ are terms and $P$, $Q$ are formulas.
\end{definition}

A list of hybrid games and their effects can be found in Table \ref{tab:dglgames}.

\begin{table}[tbhp]
    \caption{Hybrid games}
    \label{tab:dglgames}
    \begin{tabularx}{\textwidth}{|l| l |X|}
        \hline
        \textbf{Game} & \textbf{Name} & \textbf{Meaning}\\
        \hline
        $x:=e$ & Assignment game & assigns $e$ to $x$\\
        $x'=f(x)\&Q$ & Continuous game & Angel evolves ordinary differential equation (ODE) to change value of $x$ while evolution domain constraint $Q$ has to hold\\
        $?Q$ & Test game & tests if Angel fulfills $Q$, if not, she loses and Demon wins\\
        $\alpha \cup \beta$ & Choice game & Angel chooses to play either $\alpha$ or $\beta$\\
        $\alpha; \beta$ & Sequential game & $\alpha$ and $\beta$ are played sequentially\\
        $\alpha^d$ & Dual game & controls in $\alpha$ are swapped between Angel and Demon\\
        $\alpha^*$ & Repetition game & $\alpha$ is repeated finitely many times until Angel stops\\
        \hline 
    \end{tabularx}
    \end{table}

The formulas are interpreted over states. 
Each state is a function $\omega: \mathcal{V} \mapsto \mathbb{R}$ which maps all variables in the set of all variables $\mathcal{V}$ to a value in $\mathbb{R}$.
The set of all states is called $\s$. 
$\omega_x^r$ denotes the state where all variables are similar to $\omega$ except for $x$ whose value has been replaced by $r$. 
The semantics of \dGL is defined as a function $\win{\cdot}: Fml \mapsto \mathcal{P}(\s)$ which maps a formula to all the states where it is true.
For first-order formulas, the semantics is as usual.
The diamond and the box modality is defined via the functions $\varsigma_\alpha(\cdot)$ and $\delta_\alpha(\cdot)$ resp., which retrace from which states one has to start to reach a goal state in the end. 

\begin{definition}{\emph{(Semantics)}}
    The \dGL semantics is:
    \begin{itemize}
        \item $\win{\angel{\alpha}{P}} = \winrA{\alpha}{\win{P}}$
        \item $\win{\demon{\alpha}{Q}} = \winrD{\alpha}{\win{Q}}$
    \end{itemize}
    Angel's function for her winning region $\varsigma_\alpha(\cdot)$ is defined as follows:
    \begin{itemize}
        \item $\winrA{x:=e}{X} = \{\omega\in \s \;|\; \omega_x^{\omega\win{e}} \in X\}$
    \item $\winrA{x'=f(x)\&Q}{X} = \{\varphi(0) \in \s \;|\; \varphi(r)\in X \text{ for some } r \geq 0 \text{ and (differentiable) }\allowbreak \varphi:[0,r] \to  \s \text{ such that } \varphi(s) \in \win{Q} \text{ and } \frac{d\varphi(t)(x)}{dt}(s) = \varphi(s)\win{f(x)} \text{ for all }\allowbreak 0 \leq s \leq r\} \stackrel{\text{def}}{=} \{\varphi(0) \in \s \;|\; \varphi(r)\in X \text{ for some } r \geq 0 \text{ with }\varphi \models x'=f(x) \land Q\}$
        \item $\winrA{?Q}{X} = \win{Q} \cap X$
        \item $\winrA{\alpha\cup \beta}{X} = \winrA{\alpha}{X}\cup \winrA{\beta}{X}$
        \item $\winrA{\alpha;\beta}{X} = \winrA{\alpha}{\winrA{\beta}{X}}$
        \item $\winrA{\alpha^d}{X} = \winrA{\alpha}{X^\C}^\C$
        \item $\winrA{\alpha^*}{X} = \bigcap\{Z\subseteq \s\;|\; X\cup\winrA{\alpha}{Z} \subseteq Z\}$
    \end{itemize}
    Demon's winning region is defined by $\winrD{\alpha}{X} = \winrA{\alpha}{X^\C}^\C$ (\cite[Th. 3.1]{Platzer15}), i.e., Demon wins whenever Angel fails to reach the opposite of his goal, because he wins whenever Angel loses.
\end{definition}

The proof calculus for \dGL consists of all first-order logic proof rules and the rules shown in Table \ref{tab:dgl_calculus}.
There is one axiom per game construct in the diamond modality that breaks the game construct into smaller parts.
Additionally, there is another proof rule for the repetition game which allows one to get rid of the loop entirely, instead of just "unrolling" it, like the axiom for the repetition game does.
Furthermore, there is a determinacy axiom that links box and diamond modality.
Finally, the calculus also contains a monotonicity rule.

\renewcommand{\arraystretch}{1.1}
\begin{figure}[htb]
    
    \begin{tabularx}{\textwidth}{>{\color{gray}}l X >{\color{gray}}l l}
    $[\cdot]$ & $\demon{\alpha}{P} \leftrightarrow \lnot \angel{\alpha}{\lnot P}$ & $\langle \cup \rangle$ & $\angel{\alpha \cup \beta}{P} \leftrightarrow \angel{\alpha}{P} \lor \angel{\beta}{P}$ \\
    $\langle := \rangle$ & $\angel{x:=e}{p(x)} \leftrightarrow p(e)$ & $\langle {;} \rangle$ & $\angel{\alpha;\beta}{P} \leftrightarrow \angel{\alpha}{\angel{\beta}{P}}$ \\
    $\langle ' \rangle$ &  $\angel{x'=f(x)}{P} \leftrightarrow \exists t{\geq}0\, \angel{x:=y(t)}{P}$ \hskip 7pt $(y'= f(y))$ & $\langle^d\rangle$ & $\angel{\alpha^d}{P} \leftrightarrow \lnot \angel{\alpha}{\lnot P}$ \\
    $\langle?\rangle$ & $\angel{?Q}{P} \leftrightarrow Q \land P$ & $\langle^*\rangle$ & $\angel{\alpha^*}{P} \leftrightarrow P \lor \angel{\alpha}{\angel{\alpha^*}{P}}$\\
    M &
    \AxiomC{$P \rightarrow Q$}
    \UnaryInfC{$\angel{\alpha}{P} \rightarrow \angel{\alpha}{Q}$}
    \DisplayProof 
     & FP &
     \AxiomC{$P \lor \angel{\alpha}{Q} \rightarrow Q$}
    \UnaryInfC{$\angel{\alpha^*}{P} \rightarrow Q$}
    \DisplayProof
\end{tabularx}
\caption{\dGL proof calculus}
\label{tab:dgl_calculus}
\end{figure}
\section{Three-Player Differential Game logic \logic}
\label{dGL3}

    The logic \logic features three players numbered 1, 2 and 3.
    The players can work together in coalitions or play on their own.
    If a player plays on their own, they form a singleton coalition with only one member. 
    Coalitions are given as sets that contain the respective numbers of its members. 
    The set of all possible coalitions $\co$ is the powerset over the set $\{1,2,3\}$ excluding the empty set, i.e. $\co=\mathcal{P}(\{1,2,3\})\setminus \emptyset$, because coalitions without any players do not exist. 
    In a coalition, the players work together to achieve the intersection of the goals of its members.
    Players outside the coalition may play adversarially in the worst case.
    At any time, every player has to be in some coalition, which can also be the singleton coalition.

\subsection{Syntax}
    The syntax of \logic consists of first-order logic with real arithmetic and two kinds of modalities.
    The formulas are separated into state formulas and game formulas.
    To distinguish the formulas, Greek letters like $\pi, \phi$, etc. are used for  game formulas while standard Latin letters are used for state formulas.
    Game formulas \emph{cannot} stand on their own and have to be part of a state formula.
    Their special characteristic is that these game formulas are interpreted over tuples of states and coalitions, in contrast to state formulas which are interpreted over just states.
    The two modalities $\dia{C}{\pi}$ and $\bx{C}{\pi}$ are used to project the tuples from the game formulas back to states.
    These modalities are true in a state, if there exists a coalition in the specified set of coalitions $C$, resp. for all coalitions in $C$, so that the resulting state-coalition tuple is true for game formula $\pi$. 
    As the name suggests, game formulas are mainly used to describe games.
    This can be done with an additional third modality $\mo{\alpha}{\pi}$.
    Here, $\alpha$ describes the hybrid game that is played, i.e. it specifies the players' interactions,
    while $\pi$ describes the goals of the players.
    If a tuple $(\omega, c)$ of a coalition and a state is true for $\pi$, then coalition $c$ tries to reach state $\omega$ after game $\alpha$.
    The game formula $\mo{\alpha}{\pi}$ is true for a state-coalition tuple $(\omega, c)$, if coalition $c$ can reach one of its goal states after $\alpha$ starting from state $\omega$.
    The goal constructor $\proj{P_1}{P_2}{P_3}$ allows constructing goals for all possible coalitions from the players' individual goals.
    Here, $P_i$ is player $i$'s goal.
    Note that, as normal alphabet letters are used, these goals are expressed as state formulas.
    The goal of each coalition is constructed from the individual goals by taking the conjunction of the goals of the respective members of the coalition.
    State formulas can also be parsed directly to game formulas by using the $^\co$-operator on a state formula $P$, i.e. $P^\co$.
    This is the cross product of the states where $P$ is true with the set of all coalitions.
    That means $P^\co$ is true in all states where $P$ is true in combination with any coalition.
    Last but not least, game formulas can also be combined using basic propositional logic connectors.
    
    \begin{definition}[Syntax]
        The following grammar defines the \logic formulas:
        \begin{itemize}
            \item {\makebox[1.5cm][l]{$\alpha; \beta$} $::= x:=e \mid \{x'=f(x)\&Q\}_i \mid ?_iQ \mid \alpha \cup_i \beta \mid \alpha;\beta \mid \alpha^{*i}$}
            \item {\makebox[1.5cm][l]{$P_1, P_2$} $\begin{aligned}[t]
                ::= &e \geq \tilde{e} \mid \lnot P_1 \mid P_1 \land P_2 \mid \forall x P_1\\
                & \mid \exists x P_1 \mid \dia{C}{\pi_1} \mid \bx{C}{\pi_1}
            \end{aligned}$}
            \item {\makebox[1.5cm][l]{$\pi_1, \pi_2$} $::= P_1^\co \mid \pi_1 \land \pi_2 \mid \lnot \pi_1 \mid \mo{\alpha}{\pi_1} \mid \proj{P_1}{P_2}{P_3}$}
        \end{itemize}
       where  $i\in \{1,2,3\}$, $\alpha, \beta$ are hybrid games, $P_1, P_2, P_3, Q$ are state formulas, $\pi_1, \pi_2$ are game formulas, $C$ is a set of coalitions, $x$ is a variable and $e, \tilde{e}$ are terms.
    \end{definition}

The effects of the hybrid games are described in Table \ref{tab:games}.
 \begin{table}[tbhp]
    \caption{Hybrid games}
    \label{tab:games}
    \begin{tabularx}{\textwidth}{|l| l |X|}
        \hline
        \textbf{Game} & \textbf{Name} & \textbf{Meaning}\\
        \hline
        $x:=e$ & Assignment game & assigns $e$ to $x$\\
        $\{x'=f(x)\&Q\}_i$ & Continuous game & player $i$ evolves ODE to change value of $x$ while evolution domain constraint $Q$ has to hold\\
        $?_iQ$ & Test game & tests if player $i$ fulfills $Q$, if not, they lose and players outside their coalition win automatically\\
        $\alpha \cup_i \beta$ & Choice game & player $i$ chooses to play either $\alpha$ or $\beta$\\
        $\alpha; \beta$ & Sequential game & $\alpha$ and $\beta$ are played sequentially\\
        $\alpha^{*i}$ & Repetition game & $\alpha$ is repeated finitely many times until player $i$ stops\\
        \hline 
    \end{tabularx}
    \end{table}
    
  For example, in the game $(({x:=x+1} \cup_2 {x:=x-1});\{x'=-1\}_3)^{*1}$, the outermost game is a repetition game.
  The index 1 indicates that it is under control of the first player.
  They decide after each round, whether they want to play another one.
  They can also decide to play no rounds at all.
  In each loop, the second player decides whether to increment or decrement $x$ by one which can be seen by the choice operator being indexed with two.
  After that, the control is handed over to the third player which can be seen by the next game being indexed with three.
  This player then evolves the ODE $x'=-1$, i.e. time runs and $x$ changes according to the ODE.
  In this case, $x$ decreases continuously until player 3 decides to stop time to fix a new value of $x$.
  Player 3 is also allowed to stop the time after zero time has passed which leaves the value of $x$ unchanged.
  Assume now, that player 2 tries to achieve that $x>0$ and player 3 tries to achieve $x<0$ while player 1 has the trivial goal \emph{true}.
  These goals can be expressed using the goal constructor: $\proj{\top}{x>0}{x<0}$.
  Consequently, player 3 can win by teaming up with player 1.
  At the same time, player 2 and player 3 cannot win in a coalition because their goals are mutually exclusive.
  Consequently, 
  \[\dia{\{\{1,3\}\{2,3\}\}}{\mo{((x:=x+1\cup_2 x:=x-1);\{x'=-1\}_3)^{*1}}{\proj{\top}{x>0}{x<0}}}\] 
  is true in any state because some coalition from \{\{1,3\}\{2,3\}\} can win.
  On the other hand, 
  \[\bx{\{\{1,3\}\{2,3\}\}}{\mo{((x:=x+1\cup_2 x:=x-1);\{x'=-1\}_3)^{*1}}{\proj{\top}{x>0}{x<0}}}\]
   is not true because not all of the coalitions can win.

  \begin{example}
    \label{ex}
    In this little example, player 1 is car driver, player 2 is a filling station attendant and player 3 is a motorcycle rider.
    The car driver and the motorcycle rider are both out of gas.
    Unfortunately, the filling station attendant has only 5l of gas left and that is just enough for one of them to get to the next town.
    As the gas station is in the middle of nowhere, less gas would not be helpful for the drivers.
    Consequently, each driver wants to form a coalition with the filling station attendant who does not care whom he sells the gas, so his goal is \emph{true}.
    In \logic this situation can be modeled using two continuous games:
    \[
    t_c=0 \land t_m=0 \rightarrow \dia{C_1}{\mo{\{t_c'=1\}_2\cup_2 \{t_m'=1\}_2}{\proj{t_c=5}{\top}{t_m=5}}}
    \]
    At the beginning, the amount $t_c$ of gas in the car and the amount $t_m$ of gas in the motorcycle is zero.
    In the game, the filling station assistant chooses to either fill the car's or the motorcycle's tank.
    His control is shown by indexing the games with a two. 
    It should be determined whether the first player, the car driver can win the game, so the diamond modality is used with the set $C_1$ which is the abbreviation for the set of all coalitions containing player 1, i.e. $\{1\}, \{1,2\}, \{1,2,3\}$.
    This asks if there exists a coalition including player 1 that can win the game.
    This is possible, because they can win by forming a coalition with the assistant.
    The same formula, but with a box modality is \emph{not} true because not all the coalitions can win.
    Player 1 cannot win on their own because the assistant might sell the gas to the motorcycle rider.
    All players cannot win at the same time too, so coalition $\{1,2,3\}$ cannot win either.
  \end{example}

\subsection{Semantics}
This section defines the semantics of \logic formulas via a function $\win{\cdot}: \textit{Fml} \to \mathcal{P}(\s)$ which maps each state formula to the set of states where it is true.
For game formulas, the function $\win{\cdot}: \textit{Fml}_{\text{game}} \to \mathcal{P}(\s \times \co)$ returns tuples (instead of states), consisting of a state where the formula is true and a corresponding coalition.
$\s$ denotes the set of all states. 
A state is a mapping $\omega: \mathcal{V} \to \mathbb{R}$ that maps all variables to a real number.
The state $\omega^r_x$ is a state where all variables have the same values as in $\omega$, except for $x$ whose value is $r$.
$\co$ denotes the set of all coalitions.
$C_i = \{c \in \co \mid i\in c\}$ is the set of all coalitions containing player $i$.

The semantics of the first-order formulas is as usual.
The definition of the conjunction and the disjunction which appear in both state formulas and game formulas is the same, so these definitions are only given once.
The diamond modality $\dia{C}{\pi}$ is true in all states $\omega$ where there exists a coalition $c \in C$, such that $(\omega, c)$ is true for $\pi$.
The box modality $\bx{C}{\pi}$ is true in all states $\omega$ where $(\omega, c)$ is true for $\pi$ for all coalitions $c \in C$. 
The game modality $\mo{\alpha}{P}$ is true in all tuples $(\omega, c)$ where the coalition $c$ can reach their goal after playing game $\alpha$ starting from state $\omega$.
The semantics for this modality uses an additional function $\win{\alpha}(\cdot)$ with the hybrid game $\alpha$ as parameter that is defined later in this section.
The goal constructor $\proj{P_1}{P_2}{P_3}$ is true in all tuples $(\omega, c)$ where $\omega$ is in the goals of all players who are part of $c$.
$P^\co$ is true in those tuples $(\omega, c)$ where $P$ is true in $\omega$.
Formally, the semantics is defined as:
\begin{definition}[\logic semantics]
The semantics of a state formula $P$ of \logic is the subset $\win{P} \subseteq\s$ and the semantics of a game formula $\pi$ of \logic is the subset $\win{\pi} \subseteq \s \times \co$. Those subsets are defined inductively as follows:
\begin{itemize}
    \item $\win{e \geq \tilde{e}} = \{\omega \in \s \mid \omega\win{e} \geq \omega \win{\tilde{e}}\}$
    \item $\win{\lnot P_1} = \win{P_1}^\C$
    \item $\win{P_1 \land P_2} = \win{P_1}\cap\win{P_2}$
    \item $\win{\forall x P_1} = \{\omega \in \s \mid \omega_x^r \in \win{P_1} \text{ for all }r\}$
    \item $\win{\exists x P_1} = \{\omega \in \s \mid \omega_x^r \in \win{P_1} \text{ for some } r\}$
    \item $\win{\dia{C}{\pi}} = \{\omega \in \s \mid \exists c \in C: (\omega, c) \in \win{\pi}\}$
    \item $\win{\bx{C}{\pi}} = \{\omega \in \s \mid \forall c \in C: (\omega, c) \in \win{\pi}\}$
    \item $\win{P_1^\co} = \win{P} \times \co$
    \item $\win{\mo{\alpha}{\pi}} = \win{\alpha}(\win{\pi})$
    \item $\win{\proj{P_1}{P_2}{P_3}} = \{(\omega, c) \in \s \times \co \mid \omega \in \bigcap_{i\in c} \win{P_i}\}$
\end{itemize}      
\end{definition}
The function $\win{\alpha}(\cdot):\mathcal{P}(\s \times \co) \mapsto \mathcal{P}(\s \times \co)$ describes the winning region of game $\alpha$.
It takes in a set of state-coalition tuples and outputs a set of state-coalition tuples. 
A tuple $(\omega,c)$ returned by this function means that starting in state $\omega$, the coalition $c$ can win game $\alpha$ by achieving its goal. 
Goals are given to the function as input as state-coalition tuples. A tuple $(\omega, c)$ in the argument set means that the state $\omega$ is part of the goal of coalition $c$. 
The function itself is defined recursively over the hybrid game $\alpha$.
For the assignment game and the sequential game, the definition is very straightforward as both games do not involve choices and thus are not associated to any particular player.

A coalition $c$ can win the assignment game $x:=e$, if the state that is reached by replacing the value of variable $x$ by $e$, in a tuple with $c$, is part of the goal set $X$. 
In other words, a coalition can win if replacing the value of $x$ by that of $e$ leads to a state that is in the coalition's goal. Formally, this can be written as:
\[\win{x:=e}(X) = \{(\omega, c) \in \s \times C \mid (\omega_x^{\omega\win{e}}, c) \in X\}\]
To win the sequential game $\alpha;\beta$, a coalition first has to be able to win in game $\alpha$ into the winning region of game $\beta$ for the original goal $X$. From there, the coalition can then reach $X$ after game $\beta$. 
Consequently, the goal for game $\alpha$ is the winning region of game $\beta$:
\[\win{\alpha;\beta}(X) = \win{\alpha}{\big(\win{\beta}(X)\big)}\]
Unlike the previous games, the test game $?_i Q$ \emph{is} associated to a player but still does not contain any choices.
Rather, the truth of $Q$ then simply determines whether player $i$ who is responsible for passing the test $?_i Q$ does indeed pass or loses the game prematurely.
In order to pass the test, all coalitions including player $i$ have to fulfill formula $Q$ and be in their goal, as this game does not change the state. 
If $Q$ is not fulfilled, the coalition loses immediately and the rest automatically wins the game.
Consequently, all coalitions that do not include player $i$ can win if either $Q$ is not fulfilled, so the coalition of $i$ loses prematurely, or if they are in their goal:
\[\win{?_i Q}(X) = \{(\omega,c) \in \win{Q^\co} \cap  X \mid i \in c\}
\cup \{(\omega,c) \in\win{Q^\co}^\C  \cup  X \mid i \not \in c\}\]
In the choice game $\alpha\cup_i\beta$, player $i$ chooses whether to play game $\alpha$ or game $\beta$. 
All coalitions including player $i$ can win this game, if they can either win $\alpha$ or $\beta$, as they have player $i$ who can choose. 
The rest of the coalitions not including player $i$, have to be able to win both $\alpha$ and $\beta$ because they cannot rely on player $i$ choosing a beneficial option for them. 
Of course, all coalitions \emph{with} player $i$ can win as well, if they can win both games, so this part applies to all coalitions:
\[\win{\alpha\cup_i\beta}(X) = \begin{aligned}[t]
    &(\win{\alpha}(X) \cap \win{\beta}(X))\\
    & \cup (\{(\omega, c) \in \win{\alpha}(X) \mid i \in c\}\\
    & \cup \{(\omega, c) \in \win{\beta}(X) \mid i \in c\})
    \end{aligned}\]
In the continuous game $\{x'=f(x) \& Q\}_i$, player $i$ decides how long the ODE should be evolved while staying inside the \emph{evolution domain constraint} $Q$. 
This means, $i$ can evolve the ODE as long as $Q$ is true. 
If $Q$ is \emph{false} at the beginning, player $i$ and the other members of their coalition lose immediately, and all others win by default.
If player $i$ can now evolve the ODE and thus change the state to one that lies in their coalition's goal, they win.
All coalitions without player $i$ cannot rely on player $i$ stopping the evolution at an opportune moment for them.
Consequently, they need to be prepared that player $i$ stops evolving at any time. 
So to be able to win the game, they need to be in their goal for all states reachable by the ODE:
\[\begin{aligned}[t]
    &\win{\{x'=f(x)\&Q\}_i}(X)\\
    & = \{(\varphi(0), c) \in \s \times C_i \mid (\varphi(r), c) \in X
    \text{ for some } r\geq 0 \text{ with }\varphi \models x'=f(x) \land Q\}\\
    &\cup \{(\varphi(0), c) \in \s \times C_i^\C\mid (\varphi(r), c)\in X \text{ for all } r \geq 0 \text{ with } \varphi \models x'=f(x) \land Q \}
    \end{aligned}\]
The main idea for the semantics of the repetition game $\alpha^{*i}$ is that one more round of $\alpha$ should not change the winning region. 
For player $i$ and coalitions including them, this means that they can win the game either now or after another round of $\alpha$. 
Formally, this can be written as a fixpoint $Z = X \cap \win{\alpha}(Z)$. 
But this fixpoint is not unique yet.
Candidates for choosing a unique fixpoint are either the greatest or the least fixpoint.
As the empty set always fulfills this equation, the least fixpoint would not convey any information, so the greatest fixpoint is chosen instead.
For coalitions not including $i$, it should also not matter to add another round of $\alpha$.
But as they cannot decide when the loop ends, they want to stay inside their goal, no matter how long the game lasts.
Consequently, they want to stay in their goal now and after another round of $\alpha$. 
Formally, this can again be expressed using a fixpoint of the form $Z = X \cup \win{\alpha}(Z)$.
This fixpoint again, is not unique yet.
This time, the greatest fixpoint is always the whole of $\s \times \co$ which does not provide any useful information.
Therefore, the least fixpoint is chosen, since it corresponds to well-founded repetition.
The least and the greatest fixpoint can be defined by intersecting or unifying all pre-fixpoints respectively, i.e. all sets that fulfill the above fixpoint equation only with set inclusion.
Taken together, the definition for the semantics of the repetition game is:
\[\win{\alpha^{*i}}(X) = \begin{aligned}[t]
    &\bigcap \{Z \subseteq \s \times C_i \mid  X \cup \win{\alpha}(Z) \subseteq Z\}\\
    \cup &\bigcup \{Z \subseteq \s \times C_i^\C \mid X \cap \win{\alpha}(Z) \supseteq Z\}
    \end{aligned}\]

\section{Important Properties}
\label{important_properties}

In this section, some important results are discussed for \logic that are fundamental to its well-definedness or its completeness proof.
First, it is proven that the winning function of the hybrid games is monotone.
Monotonicity means that if the goals of all players are increased, the winning regions increase as well, i.e. there are more possibilities to reach a bigger goal.
This property ensures that the logic behaves in an intuitive way.
It also guarantees that the fixpoints used in the semantics exist.
The following lemma is proven by a straightforward induction over the structure of game $\alpha$:
\begin{lemmaE}[Monotonicity]
    \label{mon}
    The function $\win{\cdot}(\cdot)$ is monotone, i.e., for all hybrid games $\alpha$ and all sets $X\subseteq Y$:
    \[\win{\alpha}(X) \subseteq \win{\alpha}(Y)\]
\end{lemmaE}
\begin{proofE}
    Proof by induction over the structure of $\alpha$:
    \begin{itemize}
        \item For the assignment, test and continuous game, the claim can be directly proven by replacing $X$ with $Y$ in the definition. 
        This also applies to the repetition game.
        \item $\begin{aligned}[t]
            \win{\alpha \cup_i \beta}(X) &= \begin{aligned}[t]
                &(\win{\alpha}(X) \cap \win{\beta}(X))\\
                & \cup (\{(\omega, c) \in \win{\alpha}(X) \mid i\in c\} \cup \{(\omega, c) \in \win{\beta}(X) \mid i\in c\})
            \end{aligned}\\
            & \stackrel{\text{IH}}{\subseteq} \begin{aligned}[t]
                &(\win{\alpha}(Y) \cap \win{\beta}(Y))\\
                & \cup (\{(\omega, c) \in \win{\alpha}(Y) \mid i\in c\} \cup \{(\omega, c) \in \win{\beta}(Y) \mid i\in c\})
            \end{aligned}\\
            &= \win{\alpha \cup_i \beta}(Y)
        \end{aligned}$
        \item $\begin{aligned}[t]
            \win{\alpha;\beta}(X) &= \win{\alpha}(\win{\beta}(X))\\
            &\stackrel{\text{IH}, *}{\subseteq} \win{\alpha}(\win{\beta}(Y))\\
            &= \win{\alpha;\beta}(Y)
        \end{aligned}$

        with $*$: $\win{\beta}(X) \stackrel{\text{IH}}{\subseteq} \win{\beta}(Y)$
    \end{itemize}
\end{proofE}

In the following, it will be proven that for every \logic formula a \dGL formula with an  equivalent winning region can be constructed. 
This result will be used later on to prove the relative completeness of the proof calculus.
In order to simplify the proof, some winning region equivalences that help to break down formulas into smaller parts, are shown:
First, any modality attributed with a set of coalitions $C$ can be broken down into multiple modalities attributed with only one coalition $c \in C$.
For diamond modalities with a single coalition, one can get rid of the modality by equivalently replacing $\dia{\{c\}}{P^\co}$ with just $P$.
A goal constructor behind a diamond modality with a single coalition is true in all states where the coalition members' goals are fulfilled. 
Box modalities with a single coalition need not be considered separately: They are equivalent to diamond modalities with a single coalition.

\begin{lemmaE}[Winning region equivalences]
    \label{winequiv}
    The following winning region equivalences hold in \logic:
    \begin{enumerate}
        \item $\win{\dia{C}{\pi}} = \win{\bigvee_{c \in C}\dia{\{c\}}{\pi}}$
        \item $\win{\bx{C}{\pi}} = \win{\bigwedge_{c \in C}\bx{\{c\}}{\pi}}$
        \item $\win{\dia{\{c\}}{P^\co}} = \win{P}$
        \item $\win{\dia{\{c\}}{\proj{P_1}{P_2}{P_3}}} = \win{\bigwedge_{i\in c}P_i}$
        \item $\win{\dia{\{c\}}{\pi}} = \win{\bx{\{c\}}{\pi}}$
    \end{enumerate}
\end{lemmaE}
\begin{proofE}
    \begin{itemize}
        \item Proof for (1):
        
        "$\supseteq$": $\win{\dia{C}{\pi}} = \{\omega \in \s \mid \exists c \in C; (\omega, c) \in \win{\pi}\}$, so for any state $\omega$ in this semantics, $(\omega, c)$ is in $\win{\pi}$ for some $c\in C$. 
        Consequently, this $\omega$ is also in $\win{\dia{\{c\}}{\pi}}$. 
        No matter which exact coalition $c$ is, this is always a subset of $\bigvee_{c\in C} \dia{\{c\}}{\pi}$.
        Therefore, every $\omega$ that is in $\win{\dia{C}{\pi}}$ is also in $\win{\bigvee_{c\in C} \dia{\{c\}}{\pi}}$.

        "$\subseteq$": If a state $\omega$ is in $\win{\bigvee_{c\in C} \dia{\{c\}}{\pi}}$, then it is in the semantics of $\win{\dia{\{c\}}{\pi}} = \{\omega \in s \mid \exists \tilde{c} \in \{c\}: (\omega, \tilde{c}) \in \win{\pi}\} = \{\omega \in \s \mid (\omega, c) \in \win{\pi}\}$ for some $c \in C$. 
        Consequently, $\omega$ is also in $\{\omega \in \s \mid \exists c \in C: (\omega, c) \in \win{\pi}\} = \win{\dia{C}{\pi}}$.

        \item Proof for (2):
        
        "$\supseteq$": $\bx{C}{\pi} = \{\omega \in \s \mid \forall c \in C: (\omega, c) \in \win{\pi}\}$, so for any state $\omega$ in $\bx{C}{\pi}$, $(\omega, c)$ is in $\win{\pi}$ for all $c\in C$.
        Consequently, $\omega$ is also in $\win{\bx{\{c\}}{\pi}}$, for any $c\in C$.
        Thus, $\omega$ is also part of the intersection of all those sets.
        Therefore, any $\omega$ in $\win{\bx{C}{\pi}}$ is also in $\win{\bigwedge_{c \in C} \bx{\{c\}}{\pi}}$.
    
        "$\subseteq$": If a state $\omega$ is in $\win{\bigwedge_{c \in C} \bx{\{c\}}{\pi}}$, then it is in the semantics of $\win{\bx{\{c\}}{\pi}} = \{\omega \in \s \mid \forall \tilde{c} \in \{c\}: (\omega, \tilde{c}) \in \win{\pi}\} = \{\omega \in \s \mid (\omega,c) \in \win{\pi}\}$ for all $c \in C$.
        Consequently, $\omega$ is also in $\{\omega \in \s \mid \forall c \in C: (\omega,c) \in \win{\pi}\} = \win{\bx{C}{\pi}}$.

        \item Proof for (3):
        
        The semantics of $\dia{\{c\}}{P^\co}$ is $\{\omega \in \s \mid (\omega,c) \in \win{P} \times \co\}$. 
        As $\co$ is the set of all coalitions, a tuple $(\omega, c)$ will exist in $\win{P^\co}$, as long as $\omega$ is in $\win{P}$.
        Consequently, the set reduces to just $\win{P}$ which proves the claim.

        \item Proof for (4):
        
        The semantics of $\dia{\{c\}}{\proj{P_1}{P_2}{P_3}}$ is $\{\omega \in \s \mid (\omega,c) \in \{(\nu, d) \in \s \times \co \mid \nu \in \bigcap_{i\in d} P_i\}\}$.
        As $d$ can only be $c$, $d$ can directly be replaced by $c$.
        This simplifies the set to $\bigcap_{i\in c}P_i$ which is equal to $\win{\bigwedge_{i\in c}P_i}$ which proves the claim.

        \item Proof for (5):
        
        The semantics of $\dia{\{c\}}{\pi}$ is $\{\omega \in \s \mid \exists d \in \{c\}: (\omega, d) \in \win{\pi}\}$.
        As the set $\{c\}$ has only one element, existence is the same as something holding for all elements in the set, so this is equivalent to $\{\omega \in \s \mid \forall d \in \{c\}: (\omega, d) \in \win{\pi}\}$ which is the semantics of $\bx{\{c\}}{\pi}$.
    \end{itemize}
\end{proofE}

\begin{lemmaE}[\dGL equivalence]
    \label{dglequivalence}
    For every \logic formula $P$, a \dGL formula $P^\flat$ can be computed such that 
    \[\win{P} = \win{P^\flat}\]
\end{lemmaE}
\begin{proof}
    (Sketch) The proof is conducted by structural induction over the complexity of the formulas.
    A formula is less complex, if it contains less hybrid games, even if the overall size of the formula increases.
    First-order logic formulas need no translation, as these are contained in both logics with the same semantics.

    Formulas of the form $\dia{\{c\}}{\pi}$ and $\bx{\{c\}}{\pi}$ can be broken down into formulas only containing first-order connectives and diamond modalities with a single coalition, using the results from Lemma \ref{winequiv}.
    Consequently, it only remains to be proven, that an equivalent \dGL formula  for diamond modalities with a single coalition exists.
    The cases for $\dia{\{c\}}{\lnot \pi}$, $\dia{\{c\}}{\pi \land \psi}$, $\dia{\{c\}}{P^\co}$ and $\dia{\{c\}}{\proj{P_1}{P_2}{P_3}}$ can be proven by splitting them into smaller parts, either by definition or with the simplifications proven in Lemma \ref{winequiv}.
    Then, the IH can be used to close the proof.

    In case of the game modality, one can prove that the hybrid game can be broken down to \dGL hybrid games by grouping the players according to the coalition: 
    The players in the coalition become Angel and the rest becomes Demon, breaking the game down to a two-player game as required for \dGL.

    \begin{restatable}[Game modality equivalence]{sublemma}{todglgame}
        \label{todglgame}
        For any coalition $c\in \co$, a diamond modality for $\{c\}$ with a game modality inside can be transformed to Angel's modality in \dGL where games are redistributed such that all games that belong to players in $c$ are attributed to Angel while the rest goes to Demon and the post condition is flattened to a \dGL formula:
        \[\win{\dia{\{c\}}{\mo{\alpha}{\pi}}} = \win{\angel{\alpha\dc}{(\dia{\{c\}}{\pi})^\flat}}\]
        The exact definition of $\dc$ can be found in Table \ref{minuscdef}.
    \end{restatable}
    See proof on page \pageref{todglgameproof}.

    This lemma is proven with another induction over the structure of the game $\alpha$. 
    In this case, direct correspondence to a \dGL formula is proven, so no need to use the IH afterward, like in the previous cases.
\end{proof}
\begin{proofE}
    The claim is proven by structural induction over the complexity of the formulas. A formula is structurally less complex if it has less hybrid games, no matter how much the complexity for other connectors increases, or it is less complex, if the number of hybrid games stays the same but the number of other connectors decreases.
    \begin{itemize}
        \item $\win{P \land Q}$ is defined as $\win{P} \cap \win{Q}$. By IH, this is equal to $\win{P^\flat} \cap \win{Q^\flat}$ which is the winning region of the \dGL formula $P^\flat \land Q^\flat$. 
        Consequently, the claim is proven.
        \item Similary to the above case, the definitions for the winning regions of $\lnot P$, $\exists x P$ and $\forall x P$ break these down into winning regions of less complex formulas.
        Thus, the IH can be used, yielding the winning regions of an equivalent \dGL[] formula.
    \end{itemize}

    According to Lemma \ref{winequiv}, any \logic formula of the form $\dia{C}{\pi}$ and $\bx{C}{\pi}$, can be broken down to a formula that only contains modalities with a single coalition and first-order connectives.
    Also according to Lemma \ref{winequiv}, diamond and box modalities are equivalent, if only a single modality is attributed to them.
    Consequently, only formulas of the form $\dia{\{c\}}{\pi}$ need to be considered.
    If there is an equivalent \dGL formula for those, there is one for any \logic formula containing a diamond or a box modality.

    For the case $\dia{\{c\}}{\lnot \pi}$, the winning region is defined as $\{\omega \in \s \mid (\omega, c) \in \win{\pi}^\C\} = \{\omega \in \s \mid (\omega, c) \in \win{\pi}\}^\C$ which is the winning region of $\lnot \dia{\{c\}}{\pi}$. 
    $\dia{\{c\}}{\pi}$ is less complex than the original formula, as it contains one negation less.
    Therefore, the IH can be applied, yielding that the \dGL formula $\lnot (\dia{\{c\}}{\pi})^\flat$ has the same winning region.
    Similarly, $\dia{\{c\}}{\pi \land \psi}$ can be split at the conjunction into two less complex modalities.
    On those, the IH can be applied yielding that the \dGL formula $(\dia{\{c\}}{\pi})^\flat \land (\dia{\{c\}}{\psi})^\flat$ has an equivalent winning region.

    In the case $\dia{\{c\}}{P^\co}$, Lemma \ref{winequiv} shows that the formula is equivalent to $P$.
    This formula has one modality less, so it is structurally less complex. 
    Consequently, the IH can be applied to prove the claim.

    In the case $\dia{\{c\}}{\proj{P_1}{P_2}{P_3}}$, Lemma \ref{winequiv} shows that the formula is equivalent to $\bigwedge_{i\in c}P_i$.
    The formulas $P_i$ are less complex than the original one because they contain one modality and one goal constructor less.
    Consequently, the IH can be applied.
    As the rest of the connectives are from first-order logic, $\bigwedge_{i \in c}P_i^\flat$ is a \dGL formula with equivalent winning region.

    For the case of $\dia{\{c\}}{\mo{\alpha}{\pi}}$, it is proven that the \dGL equivalent of the formula is as follows:

    \todglgame*
    \begin{proof}
        \label{todglgameproof}
        The proof is conducted by structural induction over $\alpha$. Additional to the induction hypothesis from this structural induction, the induction hypothesis from the outer proof is used as well. To distinguish them, the outer loop induction hypothesis is named "IH1" and the induction hypothesis from this proof is named "IH2".
        \begin{itemize}
            \item For the assignment game, the definition of the semantics for $\dia{\{c\}}{\mo{x:=e}{\pi}}$ is $\{\omega \in \s \mid \exists \tilde{c} \in \{c\}: (\omega, \tilde{c}) \in \win{\mo{x:=e}{\pi}}\}$. 
            As there is only one possibility for $\tilde{c}$, this can be simplified to $\{\omega \in \s \mid (\omega, c) \in \win{\mo{x:=e}{\pi}}\}$.
            This simplification can always be done if there is only one coalition in the set.
            Therefore, it will be used directly without further notice in the future.
            The winning region of the assignment game is $\{(\omega, c) \in \s \times \co \mid (\omega_x^{\omega\win{e}}, c) \in \win{\pi}\}$.
            Inserting this in the semantics of the diamond modality yields $\{\omega \in \s \mid (\omega_x^{\omega\win{e}}, c) \in \win{\pi}\}$.
            Instead of asking whether the tuple $(\omega_x^{\omega\win{e}}, c)$ is in $\win{\pi}$, it can also be asked whether $\omega_x^{\omega\win{e}}$ is in $\win{\dia{{c}}{\pi}}$.
            That is because $\win{\dia{{c}}{\pi}}$ contains all states $\omega$ such that $(\omega, c)$ is in $\win{\pi}$.
            This step is also often reoccurring throughout the proof and will be used without further explanation.
            With this transformation, the set is now $\{\omega \in \s \mid \omega_x^{\omega\win{e}} \in \win{\dia{\{c\}}{\pi}}\}$.
            As $\dia{\{c\}}{\pi}$ contains less hybrid games than $\dia{\{c\}}{\mo{x:=e}{\pi}}$, so it is structurally less complex.
            Consequently, IH1 can be used on it:
            There exists a \dGL formula $(\dia{\{c\}}{\pi})^\flat$ with identical semantics.
            Therefore, the set can be transformed to $\{\omega \in \s \mid \omega_x^{\omega\win{e}} \in \win{(\dia{\{c\}}{\pi})^\flat}\}$ which is the semantics of the \dGL formula $\angel{x:=e}{(\dia{\{c\}}{\pi})^\flat}$.
            By definition, this equivalent to $\angel{(x:=e)\dc}{(\dia{\{c\}}{\pi})^\flat}$ which proves the claim.

            \item The definition of the semantics of $\dia{\{c\}}{\mo{?_i Q}{\pi}}$ is $\{\omega \in \s \mid (\omega,c) \in \win{\mo{?_iQ}{\pi}}\}$.
            If $i \in c$, the definition for the winning region of the test game expands to $\win{Q^\co} \cap \win{\pi}$.
            Instead of testing whether $(\omega, c)$ is in the winning region, it can also be tested whether $\omega$ is in $\win{\dia{\{c\}}{Q^\co}} \cap \win{\dia{\{c\}}{\pi}}$.
            $\win{\dia{\{c\}}{Q^\co}}$ is just $\win{Q}$.
            $Q$ and $\dia{\{c\}}{\pi}$ are less complex than the original formula, as they have one test game less, so by IH1 they have the same semantics as the \dGL formulas $Q^\flat$ and $(\dia{\{c\}}{\pi})^\flat$, respectively.
            Hence, the set can be transformed to $\win{Q^\flat} \cap \win{(\dia{\{c\}}{\pi})^\flat}$ which is equivalent to $\win{\angel{?Q^\flat}{(\dia{\{c\}}{\pi})^\flat}}$.
            By definition, this is the same as $\win{\angel{(?_iQ)\dc}{(\dia{\{c\}}{\pi})^\flat}}$ which proves the claim.

            If $i\not\in c$, the winning region for $\mo{?_i Q}{\pi}$ is $\win{Q^\co}^\C \cup \win{\pi}$.
            Similarly to the case $i\in c$, this can be transformed to $\win{Q^\flat} \cup \win{(\dia{\{c\}}{\pi})^\flat}$ which is the \dGL winning region for $\angel{(?Q^\flat)^d}{\pi^\flat}$.
            By definition, this is equivalent to $\win{\angel{(?_i Q)\dc}{(\dia{\{c\}}{\pi})^\flat}}$ which proves the claim.

            \item The definition of the semantics for $\dia{\{c\}}{\mo{\{x'=f(x) \&Q\}_i}{\pi}}$ is $\{\omega \in \s \mid (\omega,c) \in \win{\mo{\{x'=f(x)\}_i}{\pi}}\}$. 
            In the case of $i\in c$, the winning region for the continuous game reduces to $\{(\varphi(0), d) \in \s \times C_i\mid (\varphi(r), d) \in \win{\pi} \text{ for some } r\geq 0 \text{ with } \varphi \models x'=f(x) \land Q\}$.
            As only tuples are of interest where $c=d$, the semantics of the whole formula can be simplified to $\{\varphi(0) \in \s \mid \varphi(r) \in \win{\dia{\{c\}}{\pi}} \text{ for some } r \geq 0 \text{ with } \varphi \models x'=f(x) \land Q\}$.
            By IH1 $\win{\dia{\{c\}}{\pi}}$ is equivalent to $\win{(\dia{\{c\}}{\pi})^\flat}$ where $(\dia{\{c\}}{\pi})^\flat$ is a \dGL[] formula.
            Also by IH1, $\win{Q}$ is equivalent to the semantics of \dGL formula $Q^\flat$.
            Consequently, the term can be transformed to $\win{\angel{x'=f(x) \& Q^\flat}{(\dia{\{c\}}{\pi})^\flat}}$ which is equivalent to $\win{\angel{(\{x'=f(x)\&Q\}_i)\dc}{(\dia{\{c\}}{\pi})^\flat}}$ by definition.
            This proves the claim.

            In the case that $i\not\in c$, the winning region of the continuous game reduces to $\{(\varphi(0), d) \in \s \times C_i^\C \mid (\varphi(r), d) \in \win{\pi} \text{ for all } r\geq 0 \text{ with } \varphi \models x'=f(x) \land Q\}$.
            Similarly to the previous case, only coalitions where $c=d$ are of interest, so the semantics for the whole formula can be transformed to $\{\varphi(0) \in \s \mid \varphi(r) \in \win{\dia{\{c\}}{\pi}} \text{ for all } r\geq 0 \text{ with } \varphi \models x'=f(x) \land Q\}$.
            By IH1, $\win{\dia{\{c\}}{\pi}}$ is equivalent to $\win{(\dia{\{c\}}{\pi})^\flat}$, so the set is equivalent to $\win{\angel{\{x'=f(x)\&Q^\flat\}^d}{(\dia{\{c\}}{\pi})^\flat}}$.
            By definition, this can be transformed to $\win{\angel{(\{x'=f(x)\&Q\}_i)\dc}{(\dia{\{c\}}{\pi})^\flat}}$ which proves the claim.

            \item The semantics of $\dia{\{c\}}{\mo{\alpha;\beta}{\pi}}$ is $\{\omega \in \s \mid (\omega,c) \in \win{\mo{\alpha;\beta}{\pi}}\}$.
            By definition, the winning region of $\mo{\alpha;\beta}{\pi}$ is equivalent to $\win{\mo{\alpha}{\mo{\beta}{\pi}}}$.
            Now, instead of asking whether $(\omega,c) \in \win{\mo{\alpha}{\mo{\beta}{\pi}}}$, one can also ask whether $\omega \in \win{\dia{\{c\}}{\mo{\alpha}{\mo{\beta}{\pi}}}}$, i.e. the set can be transformed to $\win{\dia{\{c\}}{\mo{\alpha}{\mo{\beta}{\pi}}}}$.
            As the complexity of the hybrid games is less than in the original formula because these games contain one sequential game less, IH2 can be applied, yielding $\win{\angel{\alpha\dc}{(\dia{\{c\}}{\mo{\beta}{\pi}})^\flat}}$.
            Now, IH2 can be applied again, transforming the set to $\win{\angel{\alpha\dc}{\angel{\beta\dc}{(\dia{\{c\}}{\pi})^\flat}}}$.
            This is equivalent to the \dGL winning region of $\angel{\alpha\dc;\beta\dc}{(\dia{\{c\}}{\pi})^\flat}$.
            By definition this can then be transformed to $\win{\angel{(\alpha;\beta)\dc}{(\dia{\{c\}}{\pi})^\flat}}$ which proves the claim.

            \item The semantics for $\dia{\{c\}}{\mo{\alpha\cup_i\beta}{\pi}}$ is $\{\omega \in \s \mid (\omega,c) \in \win{\mo{\alpha\cup_i\beta}{\pi}}\}$. 
            In case $i\in c$, the winning region of the choice game can be reduced to $\win{\mo{\alpha}{\pi}} \cup \win{\mo{\beta}{\pi}}$.
            Instead of testing whether $(\omega,c)$ is in this set, one can alternatively test whether $\omega$ is in $\win{\dia{\{c\}}{\mo{\alpha}{\pi}}} \cup \win{\dia{\{c\}}{\mo{\beta}{\pi}}}$, i.e. the whole set can be simplified to just this.
            As the complexity of the hybrid games decreased, even though there are more modalities now, IH2 can be used, transforming the set to $\win{\angel{\alpha\dc}{(\dia{\{c\}}{\pi})^\flat}} \cup \win{\angel{\beta\dc}{(\dia{\{c\}}{\pi})^\flat}}$.
            Using the \dGL definition of the semantics yields $\win{\angel{\alpha\dc \cup \beta\dc}{(\dia{\{c\}}{\pi})^\flat}}$.
            With the definition of the $\dc$ operator, this can be transformed to $\win{\angel{(\alpha\cup_i\beta)\dc}{(\dia{\{c\}}{\pi})^\flat}}$ which proves the claim.

            In the case that $i\not\in c$, the definition of the choice game's winning region can be simplified to $\win{\mo{\alpha}{\pi}} \cap \win{\mo{\beta}{\pi}}$.
            As for the case above, instead of testing whether the tuple $(\omega,c)$ is in the winning region, one can also look at the states in $\win{\dia{\{c\}}{\mo{\alpha}{\pi}}} \cap \win{\dia{\{c\}}{\mo{\beta}{\pi}}}$. 
            By IH2, these can be transformed to $\win{\angel{\alpha\dc}{(\dia{\{c\}}{\pi})^\flat}} \cap \win{\angel{\beta\dc}{(\dia{\{c\}}{\pi})^\flat}}$.
            Using the definition of the \dGL semantics yields $\win{\angel{((\alpha\dc)^d\cup(\beta\dc)^d)^d}{(\dia{\{c\}}{\pi})^\flat}}$.
            By definition of the $\dc$ operator, this is equivalent to $\win{\angel{(\alpha\cup_i\beta)\dc}{(\dia{\{c\}}{\pi})^\flat}}$ which proves the claim.

            \item The semantics for $\dia{\{c\}}{\mo{\alpha^{*i}}{\pi}}$ is $\{\omega \in \s \mid (\omega, c) \in \win{\mo{\alpha^{*i}}{\pi}}\}$.
            In the case that $i \in c$, the winning region of the repetition game reduces to $\bigcap \{Z \in \s \times C_i \mid \win{\pi} \cup \win{\alpha}(Z) \subseteq Z\}$.
            As only tuples are considered where the coalition is $c$, it suffices that $Z$ is from $\s \times \{c\}$.
            Consequently, instead of testing whether $(\omega,c)$ is in the winning region, one can also check whether $\omega$ is in $\{Z \subseteq \s \mid \win{\dia{\{c\}}{\pi}} \cup \win{\dia{\{c\}}{\mo{\alpha}{\tilde{Z}^\co}}} \subseteq Z\}$ where $\tilde{Z}$ is the formula whose semantics is $Z$.
            As the second entry of the tuples in $Z$ is always the same and thus uninformative, it can be reduced away without losing information.
            Now, IH1 can be used on $\win{\dia{\{c\}}{\pi}}$ to obtain $\win{(\dia{\{c\}}{\pi})^\flat}$.
            On $\win{\dia{\{c\}}{\mo{\alpha}{\tilde{Z}}^\co}}$, IH2 can be used as $\alpha$ is a structurally less complex game than $\alpha^{*i}$, yielding $\win{\angel{\alpha\dc}{(\dia{\{c\}}{\tilde{Z}^\co})^\flat}}$.
            By Lemma \ref{winequiv}, $\dia{\{c\}}{\tilde{Z}^\co}$ is equivalent to $\tilde{Z}$.
            The winning region of $\tilde{Z}^\flat$ is equal to $\tilde{Z}$ by IH1, so the winning region of $\win{\angel{\alpha\dc}{\tilde{Z}^\flat}}$ is $\varsigma_{\alpha\dc}(Z)$.
            The set now corresponds to the winning region of the repetition game for Angel $\angel{(\alpha\dc)^*}{(\dia{\{c\}}{\pi})^\flat}$.
            By definition, this is the same as $\win{\angel{(\alpha^{*i})\dc}{(\dia{\{c\}}{\pi})^\flat}}$ which proves the claim.

            In the case of $i \not\in c$, the winning region of the repetition game simplifies to $\bigcup\{Z \subseteq \s \times C_i^\C \mid \win{\pi} \cap \win{\mo{\alpha}{Z}} \supseteq Z\}$.
            As in the previous case, $C_i^\C$ can be replaced by $\{c\}$ without losing information, so the whole set can be reduced to states instead of tuples.
            The result is $\bigcup\{Z \in \s \mid \win{\dia{\{c\}}{\pi}} \cap \win{\dia{\{c\}}{\mo{\alpha}{\tilde{Z^\co}}}} \supseteq Z\}$.
            By IH1, $\win{\dia{\{c\}}{\pi}}$ is equivalent to $\win{(\dia{\{c\}}{\pi})^\flat}$ and by IH2, $\win{\dia{\{c\}}{\mo{\alpha}{\tilde{Z^\co}}}}$ is equivalent to $\win{\angel{\alpha\dc}{(\dia{\{c\}}{\tilde{Z}^\co})^\flat}}$.
            Like in the case for $i\in c$, the latter winning region is equivalent to $\varsigma_{\alpha\dc}(Z)$.
            Consequently, the winning region corresponds to the winning region of demon's repetition game $\angel{(((\alpha\dc)^d)^*)^d}{(\dia{\{c\}}{\pi})^\flat}$.
            By definition, this is equivalent to $\win{\angel{(\alpha^{*i})\dc}{(\dia{\{c\}}{\pi})^\flat}}$  which proves the claim.
        \end{itemize}
    \end{proof}
\end{proofE}

\begin{table}[htb]
        \caption{Definition of the $\dc$ operator}
        \label{minuscdef}
        \begin{tabularx}{\textwidth}{r c l r c l}
        $(x:=e)\dc$ & $=$ & $x:=e$ & $(\{x'=f(x)\& Q\}_i)\dc$ & $=$ & $\begin{cases}
            \{x' = f(x) \& Q^\flat\} & \text{if } i \in c\\
            \{x' = f(x) \& Q^\flat\}^d & \text{else }
        \end{cases}$ \\
        $(?_i Q)\dc$ & $=$ & $\begin{cases}
            ?Q^\flat & \text{if } i \in c\\
            (?Q^\flat)^d & \text{else }
            \end{cases}$ & $(\alpha \cup_i \beta)\dc$ & $=$ & $\begin{cases}
            \alpha\dc \cup \beta\dc & \text{if } i \in c\\
            ((\alpha\dc)^d \cup (\beta\dc)^d)^d & \text{else }
        \end{cases}$ \\
        $(\alpha; \beta)\dc$ & $=$ & $\alpha\dc; \beta\dc$ & $(\alpha^{*i})\dc$ & $=$ & $\begin{cases}
            (\alpha\dc)^* & \text{if } i \in c\\
            (((\alpha\dc)^d)^*)^d & \text{else }
        \end{cases}$
    \end{tabularx}
    \end{table}

\section{Proof Calculus}
\label{proof_calculus}
For being able to make practical use of \logic, a proof calculus is introduced in this section.

The rules $\db$ and $\sdb$ define the dualities between box and diamond modality:
Axiom $\db$ specifies that pulling a negation inside a modality flips it.
This means, if no coalition in the set $C$ can win, then all of them must certainly end in the negation of the goal, i.e. they can win the game with the negated goal.
This axiom also shows that the conventional relationship between box and diamond operator holds.
Axiom $\sdb$ states that if $C=\{c\}$, box and diamond modality are equivalent, because there is only one possible coalition.

For all hybrid game constructs, there are two axioms.
These consider only diamond modalities with a single coalition associated.
More general rules can be derived, as proved in Lemma\,\ref{derivedRules}.
Depending on whether or not the coalition $c$ associated with the modality contains the player controlling the game, the game is decomposed differently into smaller parts.
For the repetition game there are two additional rules $\sFP$ and $\sind$, which allow one to get rid of the repetition game, unlike rules $\srep$ and $\srepn$ which only "unroll" the loop.
Additionally, the rule $\sMon$ expresses the monotonicity property of the game modality.
If $C$ contains both coalitions with and without the player in control, the modality can be split up using axiom $\spl$.
A corresponding axiom for the box modality can be derived from the calculus as proven in Lemma\,\ref{derived}.

Axiom $\xir$ states that the validity of the goal constructor can be proven by showing that (at least) one of the coalitions in $C$ fulfills their goal, i.e. the conjunction of its member's goals, eliminating the modality.
Another rule that allows one to get rid of the diamond modality is rule $\dco$.
This rule states that a state formula $P$ leveraged to a game formula using $^\co$ inside a modality, can be simplified to just $P$.
Rule $\dlor$ expresses that diamond modalities can be split up at disjunctions.
The corresponding axioms for the box modality can be derived using axiom $\db$.

Similar to the \dGL calculus \cite{Platzer15}, the calculus also contains a uniform substitution rule US which replaces all occurrences of predicate $q(\cdot)$ by \logic formula $Q(\cdot)$.
Note that the substitution is required to be \emph{admissible}, i.e. all variables $x$ that are replaced or occur in the replacement should not occur in the scope of a quantifier or modality binding $x$.

Additionally, the calculus includes all first-order logic rules.
The full proof calculus excluding the first-order rules and axioms, can be found in Fig. \ref{proof calculus}.

\begin{figure}[t]
    \begin{tabularx}{\textwidth}[t]{>{\color{blue}}l X >{\color{blue}}l}
        $\sas$ & $\dia{\{c\}}{\mo{x:=e}{\pi(x)}} \leftrightarrow \dia{\{c\}}{\pi(e)}$ \\
        $\ste$ & $\dia{\{c\}}{\mo{?_iQ}{\pi}} \leftrightarrow \dia{\{c\}}{Q^\co \land \pi}$ & $i \in c$\\
        $\sten$ & $\dia{\{c\}}{\mo{?_iQ}{\pi}} \leftrightarrow \dia{\{c\}}{\lnot Q^\co \lor \pi}$ & $i\not \in c$\\
        $\scon$ & $\dia{\{c\}}{\mo{\{x'=f(x)\}_i}{\pi}} \leftrightarrow \exists t{\geq} 0\, \dia{\{c\}}{\mo{x:=y(t)}{\pi}}$ & $i \in c$, $y' = f(y)$\\
        $\sconn$ & $\dia{\{c\}}{\mo{\{x'=f(x)\}_i}{\pi}} \leftrightarrow \forall t {\geq} 0\, \dia{\{c\}}{\mo{x:=y(t)}{\pi}}$ & $i\not\in c$, $y'=f(y)$\\
        $\sch$ & $\dia{\{c\}}{\mo{\alpha \cup_i \beta}{\pi}} \leftrightarrow \dia{\{c\}}{\mo{\alpha}{\pi} \lor \mo{\beta}{\pi}}$ & $i \in c$\\
        $\schn$ & $\dia{\{c\}}{\mo{\alpha \cup_i \beta}{\pi}} \leftrightarrow \dia{\{c\}}{\mo{\alpha}{\pi} \land \mo{\beta}{\pi}}$ & $i \not\in c$\\
        $\sseq$ & $\dia{\{c\}}{\mo{\alpha;\beta}{\pi}} \leftrightarrow \dia{\{c\}}{\mo{\alpha}{\mo{\beta}{\pi}}}$\\
        $\srep$ & $\dia{\{c\}}{\mo{\alpha^{*i}}{\pi}} \leftrightarrow \dia{\{c\}}{\pi \lor \mo{\alpha}{\mo{\alpha^{*i}}{\pi}}}$ & $i \in c$\\
        $\srepn$ & $\dia{\{c\}}{\mo{\alpha^{*i}}{\pi}} \leftrightarrow \dia{\{c\}}{\pi \land \mo{\alpha}{\mo{\alpha^{*i}}{\pi}}}$ & $i \not \in c$\\
        $\sFP$ & \AxiomC{$\dia{\{c\}}{\pi \lor \mo{\alpha}{\psi}} \rightarrow \psi$}
        \UnaryInfC{$\dia{\{c\}}{\mo{\alpha^{*i}}{\pi} \rightarrow \psi}$}
        \DisplayProof
        & $i\in c$\\
        $\sind$ & \AxiomC{$\dia{\{c\}}{\pi \rightarrow \mo{\alpha}{\pi}}$}
        \UnaryInfC{$\dia{\{c\}}{\pi \rightarrow \mo{\alpha^{*i}}{\pi}}$} 
        \DisplayProof
        & $i \not \in c$\\
        $\sMon$ & \AxiomC{$\dia{\{c\}}{\pi \rightarrow \psi}$}
        \UnaryInfC{$\dia{\{c\}}{\mo{\alpha}{\pi} \rightarrow \mo{\alpha}{\psi}}$}
        \DisplayProof \\
        $\spl$ &  $\dia{C \cup \{c\}}{\pi} \leftrightarrow \dia{C}{\pi} \lor \dia{\{c\}}{\pi}$ & \\
        $\xir$ & $\dia{C}{\proj{P_1}{P_2}{P_3}} \leftrightarrow \bigvee_{c \in C} \bigwedge_{i\in c} P_i$ & \\
        $\dco$ & $\dia{C}{P^\co} \leftrightarrow P$ & \\
        $\dlor$ & $\dia{C}{\pi \lor \psi} \leftrightarrow \dia{C}{\pi} \lor \dia{C}{\psi}$ & \\
        $\db$ & $\lnot \dia{C}{\lnot\pi} \leftrightarrow \bx{C}{ \pi}$ & \\
        $\sdb$ & $\dia{\{c\}}{\pi} \leftrightarrow \bx{\{c\}}{\pi}$ & \\
        US & \AxiomC{$P$}\UnaryInfC{$P_{q(\cdot)}^{Q(\cdot)}$}\DisplayProof
    \end{tabularx}
    \caption{Proof calculus for \logic}
    \label{proof calculus}
\end{figure}

One important property of any proof calculus is soundness. 
Soundness means that anything that can be proven in the proof calculus is actually valid, i.e. true in all states. 
Otherwise, the proof calculus would be practically useless.
Therefore, soundness of the \logic proof calculus is proven in the following theorem.

\begin{theoremE}[Soundness]
The \logic proof calculus is sound.
\end{theoremE}
\begin{proofE}
    The proof is done by showing that the semantics for the left side and the right side of the axioms match. 
    The soundness of the proof rules is shown by proving that the semantics of the premise is a subset of the semantics of the conclusion.
    \begin{enumerate}
        \item $\sas$:\\
        The winning region of $\mo{x:=e}{\pi(x)}$ is equal to the winning region of $\pi(e)$: 
        Instead of $\{(\omega, c) \in \s \times \co \mid (\omega_x^{\omega\win{e}}, c) \in \win{\pi(x)}\}$ one can directly write $\{(\omega, c) \in \s \times \co \mid (\omega, c) \in \win{\pi(e)}\}$, as evaluating whether $\pi(x)$ is true in a state where the value of $x$ has been replaced with the value of $e$ in state $\omega$, is the same as directly evaluating $\pi(e)$ in the state $\omega$.
        Consequently, one can replace $\win{\mo{x:=e}{\pi(x)}}$ with $\win{\pi(e)}$ in the semantics of $\dia{\{c\}}{\mo{x:=e}{\pi(x)}}$, yielding the semantics of $\dia{\{c\}}{\pi(e)}$.

        \item $\ste$ $\sten$:\\
        The semantics of $\dia{\{c\}}{\mo{?_iQ}{\pi}}$ is $\{\omega \in \s \mid (\omega, c) \in \{(\nu, d) \in \win{Q^\co} \cap \win{\pi} \mid i \in d\} \cup \{(\nu, d) \in \win{Q^\co}^\C \cup \win{\pi} \mid i \not \in d\}\}$.
        In the case that $i \in c$, the winning region of $\mo{?_i Q}{\pi} $ can be reduced to $\win{Q^\co}\cap \win{\pi}$.
        Removing the set $\{(\nu, d) \in \win{Q^\co}^\C \cup \win{\pi} \mid i \not \in d\}$ will also not change anything, as none of the tuples $(\omega, c)$ lie in this set.
        This simplification directly yields $\{\omega \in \s \mid (\omega, c) \in \win{Q^\co} \cap \win{\pi}\}$ which is the semantics of $\dia{\{c\}}{Q^\co \land \pi}$.

        The proof for axiom $\sten$ works very similar to the one for $\ste$ but as now $i\not\in c$ is assumed, the winning region can be simplified to $\{\omega \in \s \mid (\omega, c) \in \{(\nu, d) \in \win{Q}^\C \cup \win{\pi} \mid i\not \in d\}\} = \win{\dia{C}{\lnot Q \lor \pi}}$.

        \item $\scon$, $\sconn$:\\
        For axiom $\scon$, by definition, $\win{\exists t{\geq} 0\dia{\{c\}}{ \mo{x:=y(t)}{\pi}}} = \{\omega \in \s \mid \omega_t^r \in \{\nu \in \s \mid (\nu,c)\in \{(\eta, d) \in \s \times \co \mid (\eta_x^{\eta\win{y(t)}}, d) \in \win{\pi}\}\}\text{ for some }r\}$.
        Contracting these yields $\{\omega \in \s \mid ((\omega_t^r)_x^{\omega_t^r\win{y(t)}}, c) \in \win{\pi} \text{ for some }r\}$.
        Define now the function $\varphi(s) = (\omega_t^s)_x^{\omega_t^s\win{y(t)}}$ which maps a real $s$ to a state.
        This function fulfills $\varphi \models x'=f(x)$ because of the assumption that $y'=f$.
        Expanding the winning region of the other side of the axiom yields $\{\omega \in \s \mid (\omega, c) \in \{(\varphi(0), d) \in \s \times C_i \mid (\varphi(r), d) \in \win{\pi} \text{ for some } r\geq 0 \text{ with }\varphi \models x'=f(x)\} \cup \{(\varphi(0), d) \in \s \times C_i^\C \mid (\varphi(r), d) \in \win{\pi} \text{ for all } r\geq 0 \text{ with }\varphi \models x'=f(x)\}\}$.
        As known from the assumption $i \in c$, so all tuples with coalitions from $C_i^\C$ will be sorted out.
        Therefore, the set can be simplified to $\{\varphi(0) \in \s \mid  (\varphi(r), c) \in \win{\pi} \text{ for some }r \geq 0 \text{ with } \varphi\models x'=f(x)\}$. 
        It is now easy to see that the winning region of the right-hand side of the axiom is a subset of this set.
        The other inclusion holds because the solution $x'=f(x)$ is unique as proven in \cite{Platzer10}.
  
        The proof for axiom $\sconn$ works pretty similar to the one for the previous axiom.
        But as now $i \not\in c$, the winning region for the left-hand side of the axiom can be simplified to $\{\varphi(0) \in \s \mid (\varphi(r),c) \in \win{\pi} \text{ for all }r \geq 0 \text{ with } \varphi\models x'=f(x)\}$ which matches the winning region of the right-hand side of the axiom. 

        \item $\sch$, $\schn$:\\
        For axiom $\ch$, the winning region of the left-hand side is $\{\omega \in \s \mid (\omega,c) \in (\win{\alpha}(\win{\pi}) \cap \win{\beta}(\win{\pi})) \cup (\{(\nu, d) \in \win{\alpha}(\win{\pi}) \mid i \in d\} \cup \{(\nu, d) \in \win{\beta}(\win{\pi}) \mid i \in d\})\}$.
        As assumed, $i\in c$, so requiring $i \in d$ for $\{(\nu, d) \in \win{\alpha} \mid i \in d\}$ and $\{(\nu, d) \in \win{\beta} \mid i \in d\}$ is superfluous.
        Therefore, the set can be simplified to $\{\omega \in \s \mid (\omega,c) \in \win{\alpha}(\win{\pi}) \cup \win{\beta}(\win{\pi})\}$ which is the winning region of the right-hand side of the axiom.

        For the axiom $\schn$ it is assumed that $i \not\in c$. Therefore, the part $\{(\nu, d) \in \win{\alpha}(\win{\pi}) \mid i \in d\} \cup \{(\nu, d) \in \win{\beta}(\win{\pi}) \mid i \in d\}$ can be simplified away from the winning region, as $c$ does not contain $i$.
        This yields $\{\omega \in \s \mid (\omega, c) \in \win{\alpha}(\win{\pi}) \cap \win{\beta}(\win{\pi})\}$ which is the winning region of the right-hand side.

        \item $\sseq$:\\
        The axiom can be proven straightforwardly by using the definition of the semantics for the sequential game on the winning region of the left-hand side.
        This directly yields the winning region of the right-hand side.

        \item $\srep$, $\srepn$:\\
        For axiom $\srep$, the winning region of the left-hand side expands to $\{\omega \in \s \mid (\omega,c) \in \bigcap\{Z \subseteq \s \times C_i \mid \win{\pi} \cup \win{\alpha}(Z) \subseteq Z\} \cup \bigcup\{Z \subseteq \s \times C_i^\C \mid \win{\pi} \cap \win{\alpha}(Z) \supseteq Z\}\}$.
        As $i \in c$ is assumed, this can be simplified to $\{\omega \in \s \mid (\omega,c) \in \bigcup\{Z\subseteq \s \times C_i \mid \win{\pi} \cup \win{\alpha}(Z) \subseteq Z\}\}$.
        The remaining fixpoint can alternatively be described with the fixpoint equation as $\win{\pi} \cup \win{\alpha}(\win{\alpha^{*i}}(\win{\pi}))$, yielding the winning region of the right-hand side.

        The proof for axiom $\srepn$ works similar to the previous one.
        Using the assumption $i \not \in c$ now the other fixpoint can be eliminated.
        The remaining fixpoint can be rewritten, using the fixpoint property, to $\win{\pi} \cap \win{\alpha}(\win{\alpha^{*i}}(\win{\pi}))$, yielding the winning region of the right-hand side.

        \item $\sFP$:\\
        Assume that the premise  is valid. 
        Its semantics is $\win{\dia{C}{\pi \lor \mo{\alpha}{\psi}}} = \{\omega \in \s \mid (\omega, c) \in \win{\pi} \cup \win{\alpha}(\win{\psi}) \cup \win{\psi}^\C\}$.
        From the implication, it is also known that $\win{\pi} \cup \win{\alpha}(\win{\psi}) \subseteq \win{\psi}$.
        Consequently, $\win{\psi}$ is a pre-fixpoint, i.e. it fulfills the equation for the fixpoint only with inequality. 
        In this case, $\psi$ is greater than the fixpoint.
        The least fixpoint fulfilling this equation is $\bigcap\{Z\subseteq \s \times C_i \mid \win{\pi} \cup \win{\alpha}(Z) \subseteq Z\}$ (as the rule assumes $i \in c$) which is consequently also a subset of $\win{\psi}$. 
        Inside the diamond, the least fixpoint can be replaced by $\bigcap\{Z\subseteq \s \times C_i \mid \win{\pi} \cup \win{\alpha}(Z) \subseteq Z\} \cup \bigcup\{Z \subseteq \s \times C_i^\C \mid \win{\pi} \cap \win{\alpha}(Z) \supseteq Z\} = \win{\alpha^{*i}}(\pi)$, as this does not more tuples, because $c \not \in C_i^\C$.
        Therefore, $\dia{C}{\mo{\alpha^{i*}}{\pi}\rightarrow \psi}$ is also valid which proves the rule.

        \item $\sind$:\\
        Assume that the premise is valid.
        From the implication, it is known that $\win{\pi} \subseteq \win{\alpha}(\win{\pi})$.
        Intersecting both sides with $\win{\pi}$ will not change anything.
        But now, it can be seen from the result $\win{\pi} \subseteq \win{\pi} \cap \win{\alpha}(\win{\pi})$ that $\win{\pi}$ is a pre-fixpoint. 
        The greatest such fixpoint is $\bigcup\{Z \subseteq \s \times C_i^\C \mid \win{\pi} \cap \win{\alpha}(Z) \supseteq Z\}$, as $C \subseteq C_i^\C$ is assumed.
        Therefore, this fixpoint is also greater than $\win{\pi}$.
        Inside the diamond, the fixpoint can be enlarged to $\bigcap\{Z \subseteq \s \times C_i \mid \win{\pi} \cup \win{\alpha}(Z) \subseteq Z\} \cup \bigcup\{Z \subseteq \s \times C_i^\C \mid \win{\pi} \cap \win{\alpha}(Z) \supseteq Z\} = \win{\alpha^{*i}}(\win{\pi})$ without changing anything because $c \not \in C_i$.
        Consequently, the conclusion of the rule is also valid, which proves the claim.

        \item $\sMon$:\\
        This proof rule holds because of the monotonicity proven in Lemma \ref{mon}.

        \item $\spl$: \\
        The semantics of $\dia{C \cup \{c\}}{\pi}$ is $\{\omega \in \s \mid \exists d \in (C \cup \{c\}): (\omega, d) \in \win{\pi}\}$. 
        This is equivalent to asking whether either there exists a coalition in $C$ that fulfills the condition or there exists a coalition in $\{c\}$ that fulfills it, i.e. the existential quantifier can be split in two:
        \[\{\omega \in \s \mid \exists d \in C: (\omega, d) \in \win{\pi} \lor \exists d\in \{c\}: (\omega, d) \in \win{\pi}\}\]
        Now, the set can also be split in two, yielding $\{\omega \in \s \mid \exists d \in C: (\omega, d) \in \win{\pi}\} \cup \{\omega \in \s \mid \exists d \in \{c\}: (\omega, d)\in \win{\pi}\}$.
        This is exactly the semantics of $\dia{C}{\pi} \lor \dia{\{c\}}{\pi}$ which proves the axiom.
        
        \item $\xir$:\\
        The semantics for the left-hand side of the axiom is $\{\omega \in \s \mid \exists c \in C: (\omega, c) \in \{(\nu, d) \in \s \times \co \mid \omega \in \bigcap_{i \in c} \win{P_i}\}\}$.
        This can be simplified to $\{\omega \in \s \mid \exists c \in C: \omega \in \bigcap_{i\in c} \win{P_i}\}$.
        As the existential quantifier ranges over a finite set $C$, it can also be regarded as a disjunction over all possibilities: $\bigcup_{c \in C}\bigcap_{i\in c}\win{P_i}$.
        This is the semantics of the right-hand side of axiom, so the claim is proven.

        \item $\dco$:\\
        The semantics of the left-hand side of the axiom is $\{\omega \in \s \mid \exists c \in C: (\omega, c) \in \win{P} \times \co\}$. 
        Because the cross product of $\win{P}$ with $\co$ is regarded, there always exists a $c \in C$ such that $(\omega, c)$ in $\win{P} \times \co$, as long as $\omega$ is in $\win{P}$.
        Consequently, the set can be simplified to $\win{P}$ which is the semantics of the left-hand side of the axiom.

        \item $\dlor$:\\
        The semantics of the left-hand side of the axiom is $\{\omega \in \s \mid \exists c \in C: (\omega, c) \in \win{\pi} \cup \win{\psi}\}$.
        Existential quantifiers can be split at disjunctions, so this is equivalent to $\{\omega, \in \s \mid \exists c \in C: (\omega, c) \in \win{\pi} \lor \exists c \in C: (\omega, c) \in \win{\psi}\}$.
        Separating this into two sets yields $\{\omega \in \s \mid \exists c \in C: (\omega, c) \in \win{\pi}\} \cup \{\omega \in \s \mid \exists c \in C. (\omega, c) \in \win{\psi}\}$ which is the semantics of the right-hand side of the axiom.

        \item $\db$:\\
        The semantics of the left-hand side of the axiom is $\{\omega \in \s \mid \exists c \in C: (\omega, c) \in \win{\pi}^\C\}^\C$.
        To obtain the complement set, the constraint in the set is negated, yielding $\{\omega \in \s \mid \forall c \in C: (\omega, c) \in \win{\pi}\}$ which is the semantics of the right-hand side of the axiom.

        \item $\sdb$:\\
        This is a direct consequence of Lemma \ref{winequiv}.

        \item US:\\
        Standard soundness proofs for US generalize to \logic \cite{Church1956}.
        The proof presented here is adapted from \cite{Platzer15}.
        Assume that $P$ is valid, and the uniform substitution is admissible.
        Without loss of generality, assume that $q$ does not occur in $Q(\cdot)$.
        Otherwise, rename $q$ to $r$ and then use US again to replace it by $q$.
        Consider now some interpretation $J$.
        As $q$ does not occur in $P_{q(\cdot)}^{Q(\cdot)}$, the value $J(q)$ does not influence its semantics.
        Then, define interpretation $I$ which is similar to $J$, except that $\win{q(x)}^I = \win{Q(x)}^J$.
        Since $q$ does not occur in $Q(x)$, it also holds that $\win{q(x)}^I = \win{Q(x)}^I$ for all values of $x$.
        Therefore, $I \models \forall x (q(x) \leftrightarrow Q(x))$.
        For any particular occurrence $q(u)$, $I \models q(u) \leftrightarrow Q(u)$.
        Equivalents can be substituted as can be derived from the monotonicity rule, so this equivalence implies $I \models P \leftrightarrow P_{q(u)}^{Q(u)}$.
        As this is possible for any $u$ and any interpretation $J$, $I \models P \leftrightarrow P_{q(\cdot)}^{Q(\cdot)}$ is valid.

    \end{enumerate}
\end{proofE}

Until now, all rules in the proof calculus only handle modalities with a single coalition.
In practice one would rather work with sets containing more than just one coalition to avoid handling all cases separately.
Luckily, corresponding rules for sets of coalitions can be derived from the proof calculus.
So, for example, rule $\sch$ can be generalized to

\begin{tabularx}{0.7\textwidth}[t]{>{\color{blue}}l l >{\color{blue}}l}
    $\ch$ & $\dia{C}{\mo{\alpha \cup_i \beta}{\pi}} \leftrightarrow \dia{\alpha}{\pi} \cup \dia{\beta}{\pi}$ & $C \subseteq C_i$\\
\end{tabularx}\\
Instead of requiring $i$ to be in $c$, it is now required that $i$ is in any coalition $c \in C$, or in other words, $C\subseteq C_i$.
The same rule but for box modality will also be derived.
Additionally, the axioms $\aspl$ and $\asplb$ are derived from the axiom $\spl$.
These axioms break the whole set $C$ into single coalition.
These axioms will be helpful for proving relative completeness of the proof calculus later.
The rest of the rules are derivable box versions of axioms where only the version for the diamond modality is included in the proof calculus.

\begin{lemmaE}[Derived rules]
    \label{derived}
    For every axiom of the form $\dia{\{c\}}{\pi} \leftrightarrow \dia{\{c\}}{\pi'}$, the axioms
    \[\dia{C}{\pi} \leftrightarrow \dia{C}{\pi'}\]
    \[\bx{C}{\pi} \leftrightarrow \bx{C}{\pi'}\]
    are derivable.
    Generalizations of rules with a similar form are derivable too.
    Additionally, the rules

    \begin{tabularx}{0.7\textwidth}{>{\color{blue}}l X >{\color{blue}}l}
        $\splb$ & $\bx{C\cup \{c\}}{\pi} \leftrightarrow \bx{C}{\pi} \land \bx{\{c\}}{\pi}$\\
        $\aspl$ & $\dia{C}{\pi} \leftrightarrow \bigvee_{c\in C} \dia{\{c\}}{\pi}$\\
        $\asplb$ & $\bx{C}{\pi} \leftrightarrow \bigwedge_{c \in C} \bx{\{c\}}{\pi}$\\
        $\bland$ & $\bx{C}{\pi \land \psi} \leftrightarrow \bx{C}{\pi} \land \bx{C}{\psi}$\\
    \end{tabularx}\\
    can be derived.
    A full list of all derived rules and axioms can be found in Figure \ref{derivedRules} in the appendix.
\end{lemmaE}
\begin{proofE}
    \begin{enumerate}
        \item $\splb$:\\
        Using rule $\db$, $\bx{C \cup \{c\}}{\pi}$ can be transformed to $\lnot\dia{C\cup \{c\}}{\lnot \pi}$.
        By rule $\spl$, this is equivalent to $\lnot (\dia{C}{\lnot \pi}) \lor \dia{\{c\}}{\lnot \pi}$.
        Using de Morgan and $\db$ again yields $\bx{C}{\pi} \land \bx{\{c\}}{\pi}$ which proves the axiom.

        \item $\aspl$:\\
        On $\dia{C}{\pi}$ use rule $\spl$ for some $c\in C$ to obtain $\dia{C\setminus\{c\}}{\pi} \lor \dia{\{c\}}{\pi}$.
        Repeatedly use rule $\spl$, until only diamond modalities with one coalition are left.
        This yields $\bigvee{c\in C}\dia{\{c\}}{\pi}$.

        \item $\asplb$:\\
        Using rule $\db$ on $\bx{C}{\pi}$ yields $\lnot \dia{C}{\lnot \pi}$.
        This is equivalent to $\lnot (\bigvee_{c \in C} \dia{\{c\}}{\lnot \pi})$ by rule $\aspl$.
        Applying de Morgan and $\db$ again, transforms this to $\bigwedge_{c \in C} \bx{\{c\}}{\pi}$ which proves the axiom.

        \item $\xirb$:\\
        First, the rule $\asplb$ is used on $\bx{C}{\proj{P_1}{P_2}{P_3}}$, yielding $\bigwedge_{c\in C} \bx{\{c\}}{\proj{P_1}{P_2}{P_3}}$.
        By rule $\sdb$, all box modalities can be turned into diamond modalities.
        On these, one can now use rule $\xir$ to obtain $\bigwedge_{c\in C} \bigwedge_{i \in c} P_i$.

        \item $\bland$:\\
        Using rule $\asplb$, $\bx{C}{P^\co}$ can be transformed to $\bigwedge_{c \in C} \bx{\{c\}}{P^\co}$.
        By rule $\db$, this is equivalent to $\bigwedge_{c\in C} \dia{\{c\}}{P^\co}$.
        As $\dia{\{c\}}{P^\co}$ is the same as $P$ by rule $\dco$, the formula can be simplified to $P$.

        \item Axioms and rules for hybrid games:\\
        To prove axiom $\te$, first disassemble the diamond using axiom $\aspl$ to obtain $\bigvee_{c\in C} \dia{\{C\}}{\mo{?_iQ}{\pi}}$.
        Then, use rule $\ste$ on each part of the disjunction, yielding $\bigvee_{c\in C} \dia{\{c\}}{Q^\co \land \pi}$.
        The rule is applicable on each part because $C \subseteq C_i$ is assumed, i.e. $i \in c$ holds for any $c \in C$.
        After that, the parts can be reassembled to $\dia{C}{Q^\co \land \pi}$ with $\aspl$.
        All other axioms and rules can be derived similarly.
    \end{enumerate}
\end{proofE}

\begin{example}
    With the rules of the proof calculus and the derived rules, it can now be proven that the car driver from Example \ref{ex} can actually win the game and get enough gas to drive to the next town.
    See Figure \ref{proof_ex} for the proof tree.
    First, the diamond modality is split into two parts.
    One part includes all coalitions containing Player 2, the filling station assistant, while the other part contains all coalitions without them.
    This is necessary to make further axioms applicable, as Player 2 controls the game and axioms require that the set in the modality either only contains coalitions with or without Player 2.
    Then, the game is broken down further, using the choice axiom $\ch$ and the continuous axiom $\con$.
    After simplifying with axioms $\as$ and $\xir$, the proof can be closed by real arithmetic.

    \begin{figure}
        \begin{small}
            \begin{prooftree}
            \AxiomC{$*$}
            \LeftLabel{$\mathbb{R}$}
            \UnaryInfC{$t_c = 0, t_m = 0 \vdash \begin{aligned}[t]
            &\exists t {\geq} 0 ((t = 5 \land \top) \lor (t = 5 \land \top \land t_m=5))\\
            &\lor \exists t{\geq} 0 ((t_c = 5 \land \top) \lor (t_c = 5 \land \top \land t =5))
            \end{aligned} $}
            \LeftLabel{$\as$, $\xir$}
            \UnaryInfC{$t_c = 0, t_m = 0 \vdash \begin{aligned}[t]
            &\exists t{\geq}0 \dia{C_1 \cap C_2}{\mo{t_c:= t}{\proj{t_c=5}{\top}{t_m=5}}}\\
            &\lor \exists t {\geq} 0 \dia{C_1 \cap C_2}{\mo{t_m := t}{\proj{t_c=5}{\top}{t_m=5}}}
            \end{aligned}$}
            \LeftLabel{$\dlor$, $\con$}
            \UnaryInfC{$t_c = 0, t_m = 0 \vdash \begin{aligned}[t]
            &\dia{C_1 \cap C_2}{\mo{\{t_c'=1\}_2}{\proj{t_c=5}{\top}{t_m=5}}\\
            &\lor \mo{\{t_m'=1\}_2}{\proj{t_c=5}{\top}{t_m=5}}}
            \end{aligned}  $}
            \LeftLabel{wR, $\ch$}
            \UnaryInfC{$t_c = 0 , t_m = 0 \vdash \begin{aligned}[t]
            &\dia{C_1 \cap C_2}{\begin{aligned}[t]
            \mo{\{t'_c = 1\}_2 \cup_2 \{t_m' = 1\}_2\\}{\proj{t_c=5}{\top}{t_m=5}}},
            \end{aligned}\\
            &\dia{C_1 \cap C_2^\C}{\begin{aligned}[t]
            \mo{\{t'_c = 1\}_2 \cup_2 \{t_m' = 1\}_2\\}{\proj{t_c=5}{\top}{t_m=5}}
            \end{aligned}}
            \end{aligned} $}
            \LeftLabel{$\spl$, ${\lor}R$}
            \UnaryInfC{$t_c = 0, t_m = 0 \vdash \dia{C_1}{\mo{\{t'_c = 1\}_2 \cup_2 \{t_m' = 1\}_2}{\proj{t_c=5}{\top}{t_m=5}}}$}
        \end{prooftree}
        \end{small}
        \caption{Proof for Example \ref{ex}}
        \label{proof_ex}
    \end{figure}
    
    As one can see, a case distinction is only necessary to separate coalitions with and without Player 2.
    Any further case distinction is completely unnecessary.
    This saves considerable effort compared to a complete case distinction.
\end{example}

Another important property is completeness. Completeness means that any valid formula can also be proven in the proof calculus.
In the case of \logic, only relative completeness can be proven as \logic is equivalent to \dGL (Lemma \ref{dglequivalence}) which is only relatively complete.

\begin{definition}[Expressiveness \cite{Platzer15}]
    A logic $L$ is called \emph{expressive} (for \logic), if for every formula $P$, there exists an equivalent formula $P^\flat$ of $L$, i.e. $\models P \leftrightarrow P^\flat$.
    The logic $L$ is called \emph{differentially expressive}, if it is expressive and all equivalences of the form $\dia{C}{\mo{\{x'=f(x)\}_i}{\pi}} \leftrightarrow (\dia{C}{\mo{\{x'=f(x)\}_i}{\pi}})^\flat$ and $\bx{C}{\mo{\{x'=f(x)\}_i}{\pi}} \leftrightarrow (\bx{C}{\mo{\{x'=f(x)\}_i}{\pi}})^\flat$ are provable in the proof calculus.
    It is assumed that the logic $L$ is closed under first-order logic.
\end{definition}

\begin{definition}[Relative completeness \cite{Platzer15}]
    A logic is \emph{complete relative} to a differentially expressive logic $L$, if every valid formula can be proved in the calculus from $L$ tautologies.
\end{definition}

The relative completeness proof is inspired by Abou El Wafa's and Platzer's proof for equivalence of right-linear game logic and modal $\mu$-calculus \cite{DBLP:conf/lics/AbouElWafaP24}.
In order to use the relative completeness of \dGL to prove relative completeness for \logic, an equivalent translation $^\flat$ from \logic to \dGL and an equivalent translation $^\sharp$ from \dGL to \logic are defined, since \dGL and \logic do not share the same syntax.
If any \logic formula $P$ is equivalent to $P^{\flat\sharp}$, and any \dGL formula $P$ is equivalent to $P^{\sharp\flat}$, then these transformations return equivalent formulas.
As $P^{\flat\sharp}$ and $P^{\sharp\flat}$ are in the same logic as the original formulas, the equivalence to $P$ can be proven in the respective proof calculus.
In contrast to Lemma \ref{dglequivalence}, the soundness of the translation is not proven semantically but \emph{syntactically} because relative completeness is shown on a \emph{syntactical} level for the proof calculus.
\begin{lemmaE}[Provable inverses]
    \label{provable_inverses}
    \begin{enumerate}
        \item $\logic \vdash P \leftrightarrow P^{\flat\sharp}$ holds for any \logic formula $P$ and
        \item $\dGL \vdash P \leftrightarrow P^{\sharp\flat}$ holds for any \dGL formula $P$
    \end{enumerate}  

    where $^\flat$ transforms \logic formulas to \dGL formulas and is defined recursively as follows:

    \begin{tabularx}{\textwidth}{l l l}
        $p(x)^\flat = p(x)$ & $(\forall x P)^\flat = \forall x P\flat$ & $(\dia{\{c\}}{P^\co})^\flat = P^\flat$\\
        $(e\geq\tilde{e})^\flat = e\geq \tilde{e}$ & $(\exists x P)^\flat = \exists x P\flat$ & $(\dia{\{c\}}{\mo{\alpha}{\pi}})^\flat = \angel{\alpha\dc}{(\dia{\{c\}}{\pi})^\flat}$\\
        $(\lnot P)^\flat = \lnot P^\flat$ & $(\dia{\{c\}}{\lnot \pi})^\flat = \lnot (\dia{\{c\}}{\pi})^\flat$ & $(\dia{C}{\pi})^\flat = \bigvee_{c\in C} (\dia{\{c\}}{\pi})^\flat$\\
        $(P \land Q) = P^\flat \land Q^\flat$ & $(\dia{\{c\}}{\proj{P_1}{P_2}{P_3}})^\flat = \bigwedge_{i\in c} P_i^\flat$ & $(\bx{C}{\pi})^\flat = \bigwedge_{c\in C} (\dia{\{c\}}{\pi})^\flat$\\
        \multicolumn{2}{l}{$(\dia{\{c\}}{\pi_1 \land \pi_2})^\flat = (\dia{\{c\}}{\pi_1})^\flat \land (\dia{\{c\}}{\pi_2})^\flat$}
    \end{tabularx}
    The function $\dc$ transforms \logic games with three players to \dGL games with two players by grouping the players in coalition $c$ into one player.
    The same function as for the semantic equivalence proof in Lemma \ref{dglequivalence} is used which can be found in Table \ref{minuscdef}.

    The function $^\sharp$ transforms \dGL formulas to \logic formulas and is recursively defined as:
    \begin{tabularx}{\textwidth}{l l l}
        $p(x)^\flat = p(x)$ & $(\forall x P)^\sharp = \forall x P^\sharp$& $(\angel{\alpha}{P})^\sharp = \dia{\{1, 2\}}{\mo{\alpha^{+\{1,2\}}}{\proj{P^\sharp}{P^\sharp}{\lnot P^\sharp}}}$\\
        $(e\geq \tilde{e})^\sharp = e \geq \tilde{e}$ & $(\exists x P)^\sharp = \exists x P^\sharp$ & $(\demon{\alpha}{P})^\sharp = \lnot \dia{\{\{1,2\}\}}{\mo{\alpha^{+\{1,2\}}}{\proj{\lnot P^\sharp}{\lnot P^\sharp}{P^\sharp}}}$\\
        $(\lnot P)^\sharp = \lnot P^\sharp$ & $(\exists x P)^\sharp = \exists x P^\sharp$ \\
    \end{tabularx}
    To lift the two-player games to three-player games, Angel is made into a team of players 1 and 2 while Demon corresponds to player 3. 
    Games are distributed as such that all games that belong to Angel now belong to player 1 and games that belong to Demon now belong to player 3 using the function $^{+\{1,2\}}$:

    \begin{tabularx}{\textwidth}{l l}
        $(x:=e)^{+\{1,2\}} = x:=e$ & $((?Q)^d)^{+\{1,2\}} = ?_3 Q^\sharp$\\
        $(?Q)^{+\{1,2\}} = ?_1 Q^\sharp$ & $(\{x'=f(x) \& Q\}^d)^{+\{1,2\}} = \{x'=f(x)\& Q^\sharp\}_3$\\
        $(\{x' = f(x) \& Q\})^{+\{1,2\}} = \{x'=f(x) \& Q^\sharp\}_1$ & $((\alpha^d \cup \beta^d)^d)^{+\{1,2\}} = \alpha^{+\{1,2\}} \cup_3 \beta^{+\{1,2\}}$\\
        $(\alpha \cup \beta)^{+\{1,2\}} = \alpha^{+\{1,2\}} \cup_1 \beta^{+\{1,2\}}$ & $(((\alpha^d)^*)^d)^{+\{1,2\}} = (\alpha^{+\{1,2\}})^{*3}$\\
        $(\alpha^*)^{+\{1,2\}} = (\alpha^{+\{1,2\}})^{*1}$
    \end{tabularx}
   \end{lemmaE}
   \begin{proofE}
    Proof for (1):

    The claim is proven via a structural induction over the complexity of a formula.
    The complexity is defined via a rank function, where higher rank means higher complexity:
    \begin{itemize}
        \item $\text{rank}(p(x)) = 0$ for some predicate symbol $p$
        \item $\text{rank}(e \geq \tilde{e}) = 0$
        \item $\text{rank}(\lnot P) = \text{rank}(P) + 1$
        \item $\text{rank}(P \land Q) = max(\text{rank}(P), \text{rank}(Q)) +1$
        \item $\text{rank}(\forall x P) = \text{rank}(P)+1$
        \item $\text{rank}(\exists x P) = \text{rank}(P)+1$
        \item $\text{rank}(\dia{C}{\pi}) = 3|C| + \text{rank}(\pi)$
        \item $\text{rank}(\bx{C}{\pi}) = 3|C| + \text{rank}(\pi) + 1$
        \item $\text{rank}(P^\co) = \text{rank}(P)+1$
        \item $\text{rank}(\pi(x)) = 0$ for some predicate symbol $\pi$.
        \item $\text{rank}(\pi_1 \land \pi_2) = max(\text{rank}(\pi_1), \text{rank}(\pi_2))+1$
        \item $\text{rank}(\lnot \pi) = \text{rank}(\pi) +1$
        \item $\text{rank}(\proj{P_1}{P_2}{P_3}) = max(\text{rank}(P_1), \text{rank}(P_2), \text{rank}(P_3)) +1$
        \item $\text{rank}(\mo{\alpha}{\pi}) = \text{rank}(\alpha) + \text{rank}(\pi)$
        \item $\text{rank}(x:=e) = 1$
        \item $\text{rank}(?_iQ)= 1$
        \item $\text{rank}(\{x'=f(x)\& Q\}_i) = 1$
        \item $\text{rank}(\alpha;\beta) = max(\text{rank}(\alpha), \text{rank}(\beta)) +3$
        \item $\text{rank}(\alpha\cup_i\beta) = max(\text{rank}(\alpha), \text{rank}(\beta))+3$
        \item $\text{rank}(\alpha^{*i}) = \text{rank}(\alpha) +1$
    \end{itemize}
    With this now defined, the structural induction can begin:
    \begin{enumerate}
        \item The claim can be proven straightforwardly for formulas of the form $e\geq \tilde{e}$, $\lnot P$, $P \land Q$, $\forall x P$ and $\exists x P$ by using the definition of the functions $^\flat$ and $^\sharp$ and the IH.
        
        \item $(\dia{C}{\pi})^{\flat\sharp}$ with $|C| > 1$ is $\bigvee_{c \in C} (\dia{\{c\}}{\pi})^{\flat\sharp}$.
        $\dia{C}{\pi}$ can be transformed to $\bigvee_{c \in C} \dia{\{c\}}{\pi}$ with axiom $\aspl$.
        This is equivalent to  $\bigvee_{c \in C} (\dia{\{c\}}{\pi})^{\flat\sharp}$ by IH.

        The proof works similarly for $\bx{C}{\pi}$ for $|C|\geq1$ using axioms $\asplb$ and then $\sdb$ to transform all boxes into diamonds.
        As the rank has be chosen as such that boxes have a higher complexity than diamonds, it also works, if $C$ only contains one coalition, so there is no need to do case distinctions for box modalities.

        \item $\dia{\{c\}}{\lnot \pi}$ can be transformed to $\lnot \dia{\{c\}}{\pi}$ with the axioms $\db$ and $\sdb$.
        $(\dia{\{c\}}{\lnot \pi})^{\flat\sharp}$ evaluates to $\lnot (\dia{\{c\}}{\pi})^{\flat\sharp}$.
        The equivalence can now easily be proven by IH.

        \item $\dia{\{c\}}{\pi_1 \land p_2}$ can be transformed to $\dia{\{c\}}{\pi_1} \land \dia{\{c\}}{\pi_2}$ using $\sdb$, $\bland$, $\sdb$.
        $(\dia{\{c\}}{\pi_1 \land p_2})^{\flat\sharp}$ evaluates to $(\dia{\{c\}}{\pi_1})^{\flat\sharp} \land (\dia{\{c\}}{\pi_2})^{\flat\sharp}$ which is equivalent to this by IH.

        \item $\dia{\{c\}}{P^\co}$ can be transformed to $P$ using $\dco$.
        $(\dia{\{c\}}{P^\co})^{\flat\sharp}$ evaluates to $P^{\flat\sharp}$ which is equivalent to this by IH.

        \item $\dia{\{c\}}{\proj{P_1}{P_2}{P_3}}$ is equivalent to $\bigwedge_{i\in c}P_i$ by $\xir$.
        $(\dia{\{c\}}{\proj{P_1}{P_2}{P_3}})^{\flat\sharp}$ evaluates to $\bigwedge_{i\in c}P_i^{\flat\sharp}$ which is equivalent to this by IH.

        \item $\dia{\{c\}}{\mo{x:=e}{\pi(x)}}$ can be transformed to $\dia{\{c\}}{\pi(e)}$ by $\sas$.
        $(\dia{\{c\}}{\mo{x:=e}{\pi(x)}})^{\flat\sharp}$ evaluates to $\dia{\{\{1,2\}\}}{\mo{x:=e}{\proj{(\dia{\{c\}}{\pi(x)})^{\flat\sharp}}{(\dia{\{c\}}{\pi(x)})^{\flat\sharp}}{\lnot (\dia{\{c\}}{\pi(x)})^{\flat\sharp}}}}$.
        Using $\sas$, $\xir$ on this, yields $(\dia{\{c\}}{\pi(e)})^{\flat\sharp}$ which is equivalent to $\dia{\{c\}}{\pi(e)}$ by IH.

        \item $\dia{\{c\}}{\mo{?_iQ}{\pi}}$ is equivalent to $Q \land \dia{\{c\}}{\pi}$ by $\ste$, $\sdb$, $\bland$, $\sdb$ and $\dco$, if $i \in c$.
        In this case, $(\dia{\{c\}}{\mo{?_iQ}{\pi}})^{\flat\sharp}$ evaluates to $\dia{\{\{1,2\}\}}{\mo{?_1 Q^{\flat\sharp}}{\proj{(\dia{\{c\}}{\pi})^{\flat\sharp}}{(\dia{\{c\}}{\pi})^{\flat\sharp}}{\lnot(\dia{\{c\}}{\pi})^{\flat\sharp}}}}$.
        Using $\sas$, $\sdb$, $\bland$, $\sdb$, $\xir$ and $\dco$ this can be broken down to $Q^{\flat\sharp} \land \dia{\{c\}}{\pi}^{\flat\sharp}$.
        This is equivalent to $Q \land \dia{\{c\}}{\pi}$ by IH.

        In the case $i \not \in c$, $\dia{\{c\}}{\mo{?_iQ}{\pi}}$ can similarly be broken down to $\lnot Q \lor \dia{\{c\}}{\pi}$ using $\sten$ while $(\dia{\{c\}}{\mo{?_iQ}{\pi}})^{\flat\sharp}$ evaluates to $\lnot Q^{\flat\sharp} \lor (\dia{\{c\}}{\pi})^{\flat\sharp}$.
        These formulas are equivalent by IH.

        \item Similarly to the two cases above, the case $\dia{\{c\}}{\mo{\{x'=f(x)\}_i}{\pi}}$ can be proven by breaking the formulas and the version transformed with $^{\flat\sharp}$ down as much as possible.
        Then, the claim can be proven by IH.

        \item $\dia{\{c\}}{\mo{\alpha;\beta}{\pi}}$ can be broken down to $\dia{\{c\}}{\mo{\alpha}{\mo{\beta}{\pi}}}$ using $\sseq$.
        $(\dia{\{c\}}{\mo{\alpha;\beta}{\pi}})^{\flat\sharp}$ evaluates to $\dia{\{(\alpha\dc)^{+\{1,2\}}; (\beta\dc)^{+\{1,2\}}\}}{\proj{(\dia{\{c\}}{\pi})^{\flat\sharp}}{(\dia{\{c\}}{\pi})^{\flat\sharp}}{\lnot (\dia{\{c\}}{\pi})^{\flat\sharp}}}$.
        This can also be broken down using $\sseq$.
        The result \[\dia{\{\{1,2\}\}}{\mo{(\alpha)^{+\{1,2\}}}{\mo{(\beta)^{+\{1,2\}}}{\proj{(\dia{\{c\}}{\pi})^{\flat\sharp}}{(\dia{\{c\}}{\pi})^{\flat\sharp}}{\lnot (\dia{\{c\}}{\pi})^{\flat\sharp}}}}}\] is equivalent $\dia{\{\{1,2\}\}}{\mo{(\alpha\dc)^{+\{1,2\}}}{\proj{\psi}{\psi}{\lnot\psi}}}$ where $\psi$ is an abbreviation for $\dia{\{\{1,2\}\}}{\mo{(\beta)^{+\{1,2\}}}{\proj{(\dia{\{c\}}{\pi})^{\flat\sharp}}{(\dia{\{c\}}{\pi})^{\flat\sharp}}{\lnot (\dia{\{c\}}{\pi})^{\flat\sharp}}}}$.
        Now the definitions of $^\flat$ and $^\sharp$ can be used backwards to obtain $(\dia{c}{\mo{\alpha}{\mo{\beta}{\pi}}})^{\flat\sharp}$.
        By IH, this is equivalent to $\dia{c}{\mo{\alpha}{\mo{\beta}{\pi}}}$.

        \item To prove the case $\dia{\{c\}}{\mo{\alpha\cup_i \beta}{\pi}}$, both sides of the equivalence are broken down as far as possible using $\sch$ or $\schn$ respectively.
        For the resulting parts, the definitions of $^\flat$ and $^\sharp$ are used backwards.
        After that, the proof can be closed by IH.

        \item For $\dia{\{c\}}{\mo{\alpha^{*i}}{\pi}}$, both directions of the equivalence are proven separately. 
        To prove $(\dia{\{c\}}{\mo{\alpha^{*i}}{\pi}})^{\flat\sharp} \rightarrow \dia{\{c\}}{\mo{\alpha^{*i}}{\pi}}$, the definition of $^{\flat\sharp}$ is expanded first to $\dia{\{\{1,2\}\}}{\mo{((\alpha\dc)^{+\{1,2\}})^*1}{\proj{(\dia{\{c\}}{\pi})^{\flat\sharp}}{(\dia{\{c\}}{\pi})^{\flat\sharp}}{\lnot (\dia{\{c\}}{\pi})^{\flat\sharp}}}}$ in the case $i \in c$.
        Then, $\dia{\{c\}}{\mo{\alpha^{*i}}{\pi}}$ can be pulled inside the diamond modality of the left-hand side of the implication by turning it into a coalition formula with $^\co$.
        Now, the rule $\sFP$ can be used which yields $\dia{\{\{1,2\}\}}{\proj{(\dia{\{c\}}{\pi})^{\flat\sharp}}{(\dia{\{c\}}{\pi})^{\flat\sharp}}{\lnot (\dia{\{c\}}{\pi})^{\flat\sharp}} \lor \mo{(\alpha\dc)^{+\{1,2\}}}{(\dia{\{c\}}{\mo{\alpha^{*i}}{\pi}})^\co} \rightarrow (\dia{\{c\}}{\mo{\alpha^{*i}}{\pi}})^\co}$. 
        Splitting the modality with $\dlor$ and simplifying yields $(\dia{\{c\}}{\pi})^{\flat\sharp} \lor \dia{\{\{1,2\}\}}{\mo{(\alpha\dc)^{+\{1,2\}}}{(\dia{\{c\}}{\mo{\alpha^{*i}}{\pi}})^\co}} \rightarrow \dia{\{c\}}{\mo{\alpha^{*i}}{\pi}}$.
        Then, axiom $\srep$ is used on the right-hand side of the implication, transforming it to $\dia{\{c\}}{\pi} \lor \dia{\{c\}}{\mo{\alpha}{\mo{\alpha^{*i}}{\pi}}}$. 
        By monotonicity, $\dia{\{c\}}{\mo{\alpha}{\mo{\alpha^{*i}}{\pi}}}$ is equivalent to $\dia{\{c\}}{\mo{\alpha}{(\dia{\{c\}}{\mo{\alpha^{*i}}{\pi}})^\co}}$.
        With $\forall$-generalisation and US, $\dia{\{c\}}{\mo{\alpha^{*i}}{\pi}}$ is replaced by predicate $p(x)$.
        Although the additional universal quantifier increases the rank by one, replacing the repetition game by a predicate symbol lowers the rank by at least 2, so the overall complexity of the formula decreases.
        Now, $\dia{\mo{(\alpha\dc)^{+\{1,2\}}}}{\pi(x)^\co}$ can be transformed backwards to $(\dia{\{c\}}{\mo{\alpha}{p(x)^\co}})^{\flat\sharp}$ similarly to the way it is done in the case of the sequential game, using the fact that $^\flat$ and $^\sharp$ do not change predicate symbols.
        This yields $\forall x (\dia{\{c\}}{\pi})^{\flat\sharp} \lor (\dia{\{c\}}{p(x)^\co})^{\flat\sharp} \rightarrow \dia{\{c\}}{\pi} \lor \dia{\{c\}}{\mo{\alpha}{p(x)^\co}}$.
        Now, the proof can be closed by IH.
        The other direction of the proof works similarly.

        In the case of $i\not \in c$, the direction $\dia{\{c\}}{\mo{\alpha^{*i}}{\pi}} \rightarrow (\dia{\{c\}}{\mo{\alpha^{*i}}{\pi}})^{\flat\sharp}$ is proved first.
        Expanding the definition of $^{\flat\sharp}$ and merging both sides of the application into a single diamond modality yields $\dia{\{\{1,2\}\}}{(\dia{\{c\}}{\mo{\alpha^{*i}}{\pi}})^\co} \rightarrow \mo{((\alpha\dc)^{+\{1,2\}})^{*1}}{\proj{(\dia{\{c\}}{\pi})^{\flat\sharp}}{(\dia{\{c\}}{\pi})^{\flat\sharp}}{\lnot(\dia{\{c\}}{\pi})^{\flat\sharp}}}$.
        To replace the right-hand side of the implication, $\dia{\{\{1,2\}\}}{\mo{((\alpha\dc)^{+\{1,2\}})^{*1}}{(\dia{\{c\}}{\mo{\alpha^{*1}}{\pi}})^\co} \rightarrow \mo{((\alpha\dc)^{+\{1,2\}})^{*1}}{\proj{(\dia{\{c\}}{\pi})^{\flat\sharp}}{(\dia{\{c\}}{\pi})^{\flat\sharp}}{\lnot(\dia{\{c\}}{\pi})^{\flat\sharp}}}}$ is cut in.
        This implication can be proven by $\sMon$, $\srepn$ and IH.
        With this replacement, rule $\sind$ is now applicable.
        Then, axiom $\sind$ is used on the left-hand side of the application.
        Now, the statement can be proven by using $\forall$-generalisation and replacing $\dia{\{c\}}{\mo{\alpha^{*i}}{\pi}}$ by $p(x)$ with US like in the previous case of $i\in c$.
        The other direction is proven similarly.
    \end{enumerate} 

    Proof for (2):

    The proof is conducted, using an induction over the structural complexity of the formulas.
    \begin{enumerate}
        \item The validity of the claim for $e\geq\tilde{e}$, $P \land Q$, $\exists x P$ and $\forall x P$ can be easily seen by applying the definition of $^{\sharp\flat}$ and then the IH.
        \item For the cases $\angel{\alpha}{P}$ and $\demon{\alpha}{P}$, notice first, that $(\alpha^{+\{1,2\}})^{-\{1,2\}}$ is just $\alpha$.
        This can be proven by a straightforward induction over the structure of $\alpha$.
        Consequently, $(\angel{\alpha}{P})^{\sharp\flat}$ expands to $\angel{\alpha}{P^{\sharp\flat}}$.
        Its equivalence to $\angel{\alpha}{P}$ can now be proven easily by monotonicity and IH.
        Similarly, $(\demon{\alpha}{P})^{\sharp\flat}$ expands to $\lnot \angel{\alpha}{\lnot P^{\sharp\flat}}$ which is equivalent to $\demon{\alpha}{P}$ by $\langle^d\rangle$, monotonicity and IH.
    \end{enumerate}
   \end{proofE}

   As the translation is proven to be correct, it can now be shown that \logic and \dGL prove the same formulas.
   \begin{lemmaE}[Equipotency]
    \label{equipotency}
    Formulas that can be proved in \dGL can also be proved in \logic and vice versa:
    \begin{enumerate}
        \item \dGL $\vdash P$ iff \logic $\vdash P^\sharp$ for any \dGL formula $P$
        \item \logic $\vdash P$ iff \dGL $\vdash P^\flat$ for any \logic formula $P$
    \end{enumerate}
   \end{lemmaE}
   \begin{proofE}
    Proof for (1):

    The proof for the direction \dGL $\vdash P$ implies \logic $\vdash P^\sharp$ is conducted by an induction over the structure of the proof.
    \begin{enumerate}
        \item If the proof for $P$ is closed using axiom $\langle := \rangle$, $\angel{x:=e}{p(x)} \leftrightarrow p(e)$, the same equivalence can be derived in \logic by applying $\sas$ and $\xir$ to $(\angel{x:=e}{p(x)})^\sharp \leftrightarrow \dia{\{\{1,2\}\}}{\mo{x:=e}{\proj{p(x)^\sharp}{p(x)^\sharp}{\lnot p(x)^\sharp}}}$.
        \item Equivalences for the axioms $\langle'\rangle$, $\langle ? \rangle$, $\langle \cup \rangle$, $\langle ; \rangle$ and $\langle^* \rangle$ can be derived similarly by using the correspondent axiom in \logic, then bringing the result in the right shape with $\xir$, $\dco$ and the definition of $^\sharp$.
        \item Because there is no uniform handling of the dual game, all cases for $\alpha$ in axiom $\langle^d \rangle$ need to be considered separately via an induction over the structure of the game.
        
        For the assignment game, the test game and the continuous game, the equivalence can be derived by breaking the game down with the respective axiom and inserting a  double negation.
        In the case of the test game and the continuous game, the axiom for the case $i\in c$ can be used to reassemble the game and proving the equivalence.
        For the assignment game, simply the same axiom is used.

        The proof works similarly for the choice game, but the IH needs to be used only once on the parts, in which the game is broken into.

        For the sequential game, $\sseq$ is used first to break it down.
        Then, the result $\dia{\{\{1,2\}\}}{\mo{(\alpha^d)^{+\{1,2\}}}{\mo{(\beta^d)^{\{1,2\}}}{\proj{P^\sharp}{P^\sharp}{\lnot P^\sharp}}}}$ is transformed to \[\begin{aligned}
        \dia{\{\{1,2\}\}}{\mo{(\alpha^d)^{\{1,2\}}}{\proj{\dia{\{\{1,2\}\}}{\mo{(\beta^d)^{\{1,2\}}}{\proj{P^\sharp}{P^\sharp}{\lnot P^\sharp}}}\\}{\dia{\{\{1,2\}\}}{\mo{(\beta^d)^{\{1,2\}}}{\proj{P^\sharp}{P^\sharp}{\lnot P^\sharp}}}}{\lnot \dia{\{\{1,2\}\}}{\mo{(\beta^d)^{\{1,2\}}}{\proj{P^\sharp}{P^\sharp}{\lnot P^\sharp}}}}}}
        \end{aligned}\]
        The equivalence of this transformation can easily be seen by monotonicity.
        Now, the IH can be used twice, and the result can be reassembled using $\sseq$ again, proving the claim.

        In the case of the repetition game, both directions of the equivalence are proved separately. 
        The implication \[\dia{\{\{1,2\}\}}{\mo{(\alpha^d)^{*3}}{\proj{P^\sharp}{P^\sharp}{\lnot P^\sharp}}} \rightarrow \lnot \dia{\{\{1,2\}\}}{\mo{(\alpha^d)^{*1}}{\proj{\lnot P^\sharp}{\lnot P^\sharp}{P^\sharp}}}\] needs to be inverted to \[\dia{\{\{1,2\}\}}{\mo{(\alpha^d)^{*1}}{\proj{\lnot P^\sharp}{\lnot P^\sharp}{P^\sharp}}} \rightarrow \lnot \dia{\{\{1,2\}\}}{\mo{(\alpha^d)^{*3}}{\proj{P^\sharp}{P^\sharp}{\lnot P^\sharp}}}\] first, then it can be proven using $\sFP$ and $\srepn$.
        The other direction can be proven using $\sind$ and $\srep$.

        \item The equivalent to rule $[\cdot]$ follows directly from the definition of the $^\flat$ operator.
        
        \item The equivalent for the monotonicity rule can be derived by using $\sMon$, and then simplifying the result with $\xir$.
        
        \item To derive the FP rule, the $Q^\sharp$ in $\dia{\{\{1,2\}\}}{\mo{(\alpha^{+\{1,2\}})}{\proj{P^\sharp}{P^\sharp}{\lnot P^\sharp}}} \rightarrow Q^\sharp$ is pulled into the diamond by rewriting it to $\proj{Q^\sharp}{Q^\sharp}{\lnot Q^\sharp}$ as $\dia{\{\{1,2\}\}}{\proj{Q^\sharp}{Q^\sharp}{\lnot Q^\sharp}}$ simplifies to just $Q^\sharp$.
        Now, the $\sFP$ rule is applicable.
        Splitting up the diamond modality and simplifying with $\xir$ then yields $(P \lor \angel{\alpha}{Q} \rightarrow Q)^\sharp$.

        \item First-order rules are included in both \dGL and \logic calculus, so if a derivation ends with a first-order rule, the same rule can be applied in \logic because the translation does not alter first-order connectives.
    \end{enumerate}

    Now, the other direction, \logic $\vdash P^\sharp$ implies \dGL $\vdash P$, is proved.
    As for the other direction before, this is done via a structural induction over the proof, considering the last rule used to close the proof.
    If an axiom or rule for $i \in c$ or monotonicity is used, an equivalent derivation can be achieved with the corresponding \dGL axiom.
    The axioms for $i \not \in c$ can be derived using axiom $\langle^d \rangle$.
    For axioms $\xir$, $\dco$ and $\dlor$ both sides are the same in \dGL.
    Therefore, the equivalent derivation in \dGL would be no transformation at all.
    $\spl$ will never be applied because the translation $^\sharp$ does not generate diamond modalities with more than one coalition.
    Same for rules $\db$ and $\sdb$ because the translation does not generate boxes.
    Like in the previous direction, first-order rules are included in both calculi, so the exact same rule can be used.

    Proof for (2):\\
    The proof is conducted in both directions via an induction over the structure of the proof.
    First, the direction \logic $\vdash P$ implies \dGL $\vdash P^\flat$:
    All rules and axioms for $i \in c$ and monotonicity have a directly corresponding axiom in \dGL, so this axiom can be used for an equivalent transformation.
    For derivations made with an axiom for $i \not \in c$, an equivalent derivation can be achieved in \dGL by additionally using the $\langle^d\rangle$ axiom.
    The equivalent of rules $\spl$, $\xir$, $\dco$, $\dlor$, $\db$ and $\sdb$ is just the identity because both sides are the same when translated to \dGL.

    Now, the direction \dGL $\vdash P^\flat$ implies \logic $\vdash P$ is proved.
    Rules $[\cdot]$ and $\langle^d\rangle$ will never be applied in the derivation because the translation $^\flat$ neither generates box modalities nor dual games.
    For all other rules and axioms the corresponding axiom in \logic can be applied to achieve an equivalent transformation.
   \end{proofE}

With this result relative completeness of \dGL directly transfers to \logic.

\begin{theorem}[Relative completeness]
    The proof calculus for \logic is complete relative to any differentially expressive logic $L$.   
\end{theorem}
\begin{proof}
   Relative completeness of \logic to any differentially expressive logic $L$ can be proved via the relative completeness of \dGL.
   Any valid formula $P$ can be proven in \dGL from $L$ tautologies, i.e. \dGL $\vdash_L P$ \cite[Th. 4.5]{Platzer15}.
   As anything that can be proven in \dGL, can also be proven in \logic (L. \ref{equipotency}), the translated formula $P^\sharp$ can also be proven in \logic, i.e. \logic $\vdash_L P^\sharp$.
   Additionally, if  $P$ is a valid formula in \logic, $P^\flat$ is also valid in \dGL because of the equivalence of the translation (L. \ref{provable_inverses}).
   As \dGL is relatively complete, $P^\flat$ can be proven from $L$ tautologies, i.e. \dGL $\vdash_L P^\flat$.
   Consequently, \logic $\vdash_L P$ also holds due to L. \ref{equipotency}.
\end{proof}

\section{Conclusion}    

This paper presented the logic \logic for reasoning about three-player non-zero-sum hybrid games.
\logic makes it possible to handle these with hybrid dynamics while allowing the formation of coalitions between players at will during game play without fixing them statically beforehand (which would remove coalition forming entirely).

A syntax, a semantics and a sound and relatively complete proof calculus for \logic are introduced in this paper.
Additionally, monotonicity of the game modality has been shown and an equivalent transformation to \dGL has been proved.
The latter result is then used to prove relative completeness for the proof calculus, thus, showing that the calculus is adequately defined.

In the future, \logic can be extended to $n$-player games, as those have the same essential complexity as three-player games.
This would allow for more comprehensive proofs of scenarios involving multiple agents.
Furthermore, it is interesting to implement \logic in a theorem prover and conduct case studies.

\bibliography{refs}{}
\section{Appendix}
\begin{figure}
    \begin{tabularx}{\textwidth}{>{\color{blue}}l X >{\color{blue}}l}
        $\splb$ & $\bx{C\cup \{c\}}{\pi} \leftrightarrow \bx{C}{\pi} \land \bx{\{c\}}{\pi}$\\
        $\aspl$ & $\dia{C}{\pi} \leftrightarrow \bigvee_{c\in C} \dia{\{c\}}{\pi}$\\
        $\asplb$ & $\bx{C}{\pi} \leftrightarrow \bigwedge_{c \in C} \bx{\{c\}}{\pi}$\\
        $\xirb$ & $\bx{C}{\proj{P_1}{P_2}{P_3}} \leftrightarrow \bigwedge_{c\in C} \bigwedge_{i \in c} P_i$\\
        $\dcob$ & $\bx{C}{P^\co} \leftrightarrow P$\\
        $\bland$ & $\bx{C}{\pi \land \psi} \leftrightarrow \bx{C}{\pi} \land \bx{C}{\psi}$\\
        $\Mon$ & \AxiomC{$\dia{C}{\pi \rightarrow \psi}$} \UnaryInfC{$\dia{C}{\mo{\alpha}{\pi} \rightarrow \mo{\alpha}{\psi}}$} \DisplayProof & \\[10pt]
        $\as$ & $\dia{C}{\mo{x:=e}{\pi(x)}} \leftrightarrow \dia{C}{\pi(e)}$ & \\
        $\te$ & $\dia{C}{\mo{?_i Q}{\pi}} \leftrightarrow \dia{C}{Q^\co \land \pi}$ & $C \subseteq C_i$ \\
        $\ten$ & $\dia{C}{\mo{?_i Q}{\pi}} \leftrightarrow \dia{C}{\lnot Q^\co \lor \pi}$ & $C \subseteq C_i^\C$\\
        $\con$ & $\dia{C}{\mo{\{x'=f(x)\}_i}{\pi}} \leftrightarrow \exists t {\geq} 0: \dia{C}{\mo{x:=y(t)}{\pi}}$ & $C \subseteq C_i$, $y'=f$ \\
        $\conn$ & $\dia{C}{\mo{\{x'=f(x)\}_i}{\pi}} \leftrightarrow \forall t {\geq} 0:\dia{C}{\mo{x:=y(t)}{\pi}}$ & $C \subseteq C_i^\C$, $y'=f$ \\
        $\ch$ & $\dia{C}{\mo{\alpha \cup_i \beta}{\pi}} \leftrightarrow \dia{C}{\mo{\alpha}{\pi} \lor \mo{\beta}{\pi}}$ & $C \subseteq C_i$\\
        $\chn$ & $\dia{C}{\mo{\alpha \cup_i \beta}{\pi}} \leftrightarrow \dia{C}{\mo{\alpha}{\pi} \land \mo{\beta}{\pi}}$ & $C \subseteq C_i^\C$\\
        $\seq$ & $\dia{C}{\mo{\alpha; \beta}{\pi}} \leftrightarrow \dia{C}{\mo{\alpha}{\mo{\beta}{\pi}}}$ & \\
        $\rep$ & $\dia{C}{\mo{\alpha^{*i}}{\pi}} \leftrightarrow \dia{C}{\pi \lor \mo{\alpha}{\mo{\alpha^{*i}}{\pi}}}$ & $C \subseteq C_i$\\
        $\repn$ & $\dia{C}{\mo{\alpha^{*i}}{\pi}} \leftrightarrow \dia{C}{\pi \land \mo{\alpha}{\mo{\alpha^{*i}}{\pi}}}$ & $C \subseteq C_i^\C$\\[2pt]
        $\FP$ & \AxiomC{$\dia{C}{\pi \lor \mo{\alpha}{\psi} \rightarrow \psi}$}
        \UnaryInfC{$\dia{C}{\mo{\alpha^{*i}}{\pi} \rightarrow \psi}$}
        \DisplayProof & $C \subseteq C_i$\\[10pt]
        $\ind$ & \AxiomC{$\dia{C}{\pi \rightarrow \mo{\alpha}{\pi}}$}
        \UnaryInfC{$\dia{C}{\pi \rightarrow \mo{\alpha^{*i}}{\pi}}$}
        \DisplayProof & $C \subseteq C_i^\C$ \\
        $\asb$ & $\bx{C}{\mo{x:=e}{\pi(x)}} \leftrightarrow \bx{C}{\pi(e)}$ & \\
        $\tenb$ & $\bx{C}{\mo{?_i Q}{\pi}} \leftrightarrow \bx{C}{Q \land P}$ & $C \subseteq C_i$ \\
        $\tenb$ & $\bx{C}{\mo{?_i Q}{\pi}} \leftrightarrow \bx{C}{\lnot Q \lor \pi}$ & $C \subseteq C_i^\C$\\
        $\conb$ & $\bx{C}{\mo{\{x'=f(x)\}_i}{\pi}} \leftrightarrow \exists t {\geq} 0: \bx{C}{\mo{x:=y(t)}{\pi}}$ & $C \subseteq C_i$, $y'=f$ \\
        $\connb$ & $\bx{C}{\mo{\{x'=f(x)\}_i}{\pi}} \leftrightarrow \forall t {\geq} 0:\bx{C}{\mo{x:=y(t)}{\pi}}$ & $C \subseteq C_i^\C$, $y'=f$  \\
        $\chb$ & $\bx{C}{\mo{\alpha \cup_i \beta}{\pi}} \leftrightarrow \bx{C}{\mo{\alpha}{\pi} \lor \mo{\beta}{\pi}}$ & $C \subseteq C_i$\\
        $\chnb$ & $\bx{C}{\mo{\alpha \cup_i \beta}{\pi}} \leftrightarrow \bx{C}{\mo{\alpha}{\pi} \land \mo{\beta}{\pi}}$ & $C \subseteq C_i^\C$\\
        $\seqb$ & $\bx{C}{\mo{\alpha; \beta}{\pi}} \leftrightarrow \bx{C}{\mo{\alpha}{\mo{\beta}{\pi}}}$ & \\
        $\repb$ & $\bx{C}{\mo{\alpha^{*i}}{\pi}} \leftrightarrow \bx{C}{\pi \lor \mo{\alpha}{\mo{\alpha^{*i}}{\pi}}}$ & $C \subseteq C_i$\\
        $\repnb$ & $\bx{C}{\mo{\alpha^{*i}}{\pi}} \leftrightarrow \bx{C}{\pi \land \mo{\alpha}{\mo{\alpha^{*i}}{\pi}}}$ & $C \subseteq C_i^\C$\\[2pt]
        $\FPb$ & \AxiomC{$\bx{C}{\pi \lor \mo{\alpha}{\psi} \rightarrow \psi}$}
        \UnaryInfC{$\bx{C}{\mo{\alpha^{*i}}{\pi} \rightarrow \psi}$}
        \DisplayProof & $C \subseteq C_i$\\[10pt]
        $\indb$ & \AxiomC{$\bx{C}{\pi \rightarrow \mo{\alpha}{\pi}}$}
        \UnaryInfC{$\bx{C}{\pi \rightarrow \mo{\alpha^{*i}}{\pi}}$}
        \DisplayProof & $C \subseteq C_i^\C$ \\
        $\Monb$ & \AxiomC{$\bx{C}{\pi \rightarrow \psi}$} \UnaryInfC{$\bx{C}{\mo{\alpha}{\pi} \rightarrow \mo{\alpha}{\psi}}$}\DisplayProof & \\
    \end{tabularx}
    \caption{Derived rules}
    \label{derivedRules}
\end{figure}
\printProofs
\end{document}